\documentclass[reqno,oneside,12pt]{amsart}

\usepackage{amssymb,graphicx,braket,amsmath,amsfonts,amsthm}
\usepackage{hyperref}
\usepackage[utf8]{inputenc}
\usepackage{braket}
\usepackage{color}
\usepackage{amsaddr}
\usepackage{tikz}
\usetikzlibrary{shapes,arrows,er,positioning}
\usetikzlibrary{shapes.multipart}
\usepackage{stmaryrd}
\usepackage{fancyhdr}
\usepackage{array}
\usepackage{wasysym}
\usetikzlibrary{arrows}
\allowdisplaybreaks

\theoremstyle{plain}
\newtheorem{theorem}{Theorem}[section]
\theoremstyle{plain}
\newtheorem{proposition}[theorem]{Proposition}
\newtheorem*{proposition*}{Proposition}
\newtheorem{remark}[theorem]{Remark}
\newtheorem*{theorem*}{Theorem}

\newtheorem{lemma}[theorem]{Lemma}

\newtheorem*{assumption*}{Assumption}
\newtheorem{corollary}[theorem]{Corollary}

\newtheorem{conjecture}[theorem]{Conjecture}

\theoremstyle{definition} 
\newtheorem{definition}[theorem]{Definition}

\numberwithin{equation}{section}

\def\mainmatter{\def\baselinestretch{1.1}\normalfont}

\usepackage[margin=1in]{geometry}

\makeatletter
\renewcommand{\section}{\@startsection
{section}
{1}
{\z@}
{-\baselineskip}
{0.8\baselineskip}
{\centering\scshape\large}} 

\renewcommand{\subsection}{\@startsection
{subsection}
{2}
{\z@}
{-0.8\baselineskip}
{0.5\baselineskip}
{\normalfont \bf \normalsize}} 

\renewcommand{\subsubsection}{\@startsection
{subsubsection}
{3}
{\z@}
{-0.8\baselineskip}
{0.5\baselineskip}
{\normalfont \it \normalsize}} 
\makeatother

\newcommand{\arXiv}[1]{\href{http://arxiv.org/abs/#1}{\texttt{arXiv\string:\allowbreak#1}}}

\begin{document}
\title{Super Topological Recursion and Gaiotto Vectors For Superconformal Blocks}

\author{Kento Osuga$^{1,2}$}
\address{$^1$School of Mathematics and Statistics, University of Sheffield, Hicks Building, Hounsfield Road, Sheffield, S3 7RH, United Kingdom\\
$^2$Graduate School of Mathematical Sciences, University of Tokyo, 3-8-1 Komaba, Meguro, Tokyo, 153-8914, Japan}\email{osuga@ms.u-tokyo.ac.jp}

\begin{abstract}
We investigate a relation between the super topological recursion and Gaiotto vectors for $\mathcal{N}=1$ superconformal blocks. Concretely, we  introduce the notion of  the untwisted and $\mu$-twisted super topological recursion, and construct a dual algebraic description in terms of super Airy structures.  We then show that the partition function of an appropriate super Airy structure coincides with the Gaiotto vector for $\mathcal{N}=1$ superconformal blocks in the Neveu-Schwarz or Ramond sector. Equivalently, the Gaiotto vector can be computed by the untwisted or $\mu$-twisted super topological recursion. This implies that the framework of the super topological recursion--  equivalently super Airy structures -- can be applied to compute the Nekrasov partition function of $\mathcal{N}=2$ pure $U(2)$ supersymmetric gauge theory on $\mathbb{C}^2/\mathbb{Z}_2$ via a conjectural extension of the Alday-Gaiotto-Tachikawa correspondence.

\end{abstract}

\newpage\maketitle
\setcounter{tocdepth}{1}\tableofcontents\mainmatter

\newpage\section{Introduction}\label{sec:intro}

\subsection{Backgrounds}
The Alday-Gaiotto-Tachikawa (AGT) correspondence conjectures an intriguing relation between conformal field theories in two dimensions and $\mathcal{N}=2$ supersymmetric gauge theories in four dimensions. It was conjectured in \cite{AGT} that conformal blocks of the Liouville field theory coincide with the Nekrasov partition function of $\mathcal{N}=2$ supersymmetric gauge theory of gauge group $SU(2)$ in four dimensions \cite{N}. Since then, this correspondence has been intensively studied to verify the details and to extend with great generalities. See \cite{GNY,LF,NY1,NY2,NO,T1,T2} and references therein for reviews and recent developments of the AGT correspondence\footnote{The AGT correspondence can be, in fact, thought of a paricular example of a wider class of duality, the so-called BPS/CFT correspondence reviewed in \cite{N1,N2,N3,N4,N5} which describes a relation between supersymmetric gauge theory, integrable systems, and 2d conformal field theory. Another dscendant of the BPS/CFT correspondence is the so-called Bethe/Gauge correspondence \cite{NS1,NS2,NS3} which indeed predates the AGT correspondence.}.

For $SU(2)$ theories (with or without matters) on $\mathbb{C}^2$, it was conjectured by Gaiotto \cite{G} that the Nekrasov partition function equals to the norm of the Whittaker vector in the corresponding Verma module of the Virasoro algebra, which is often called the ``Gaiotto vector'' in literature. This is a special example of the AGT correspondence because the Gaiotto vector can be realised as the irregular limit of a four-punctured conformal block. The relation between Nekrasov partition functions and Gaiotto vectors has been generalised in several ways, and in particular, it is proven \cite[Theorem 1.4.1]{BFN} that the Nekrasov partition function $Z_{\text{Nek}}^{\mathcal{G}\text{, pure}}$  of pure theory of any simply-laced gauge group $\mathcal{G}$ coincides with the norm of the Gaiotto vector $\ket{G_{\mathfrak{g}}}$ for the corresponding $\mathcal{W}(\mathfrak{g})$-algebra. 

In a completely different context, the notion of ``Airy structures'' was introduced by Kontsevich and Soibelman \cite{KS} (see also \cite{ABCD}) as an algebraic reformulation (and generalisation) of the Chekhov-Eynard-Orantin (CEO) topological recursion \cite{CEO,EO,EO2}. In short, an Airy structure is a collection of differential operators $\{H_i\}_{i\in\mathbb{Z}_{>0}}$ with certain properties, and there exists a unique formal power series solution $Z_{\text{Airy}}$ satisfying differential equations $H_i Z_{\text{Airy}}=0$ from which $Z_{\text{Airy}}$ is recursively determined. Shortly after, Borot et al \cite{HAS} found a systematic recipe to construct an Airy structure from a \emph{twisted} module of the $\mathcal{W}(\mathfrak{gl}_r)$-algebra\footnote{They found Airy structures for $\mathcal{W}(\mathfrak{g})$-algebras of type D or E as well}, and showed an equivalence to the $r$-ramified Bouchard-Eynard topological recursion \cite{BE} on an appropriate spectral curve. The representation of the differential operators $\{H_i\}^{\mathfrak{gl}_r}_{i\in\mathbb{Z}_{>0}}$ is in a one-to-one correspondence with the defining data of the spectral curve, and the partition function $Z^{\mathfrak{gl}_r}_{\text{Airy}}$ encodes the same information as the multilinear differentials $\omega_{g,n}$ obtained by the Bouchard-Eynard topological recursion.

Since $\mathcal{W}$-algebras play an important role in the AGT correspondence, Airy structures, and the topological recursion, is there any relation between them?  This point was recently addressed by Borot, Bouchard, Chidambaram, and Creutzig (BBCC) \cite{BBCC}. They found an Airy structure as an \emph{untwisted} module of a $\mathcal{W}(\mathfrak{g})$-algebra of type A, B, C or D whose partition function $Z_{\text{Airy}}^{\mathfrak{g}}$ is none other than the Gaiotto vector $\ket{G_{\mathfrak{g}}}$ expressed as a power series of formal variables. With an appropriate inner product $(\cdot|\cdot)$ for $Z_{\text{Airy}}^{\mathfrak{g}}$, the AGT correspondence and the results of BBCC imply that 
\begin{equation}
Z_{\text{Nek}}^{\mathcal{G}\text{, pure}} \;\overset{\text{AGT}}{=}\;  \braket{G_{\mathfrak{g}}|G_{\mathfrak{g}}} \;\overset{\text{BBCC}}{=}\; (Z_{\text{Airy}}^{\mathfrak{g}}|Z_{\text{Airy}}^{\mathfrak{g}}), \label{BBCC}
\end{equation}
where the first equality is not proven for type B and C, but it is expected to hold.

However, it is worth emphasizing that the topological recursion dual to the Airy structure of \cite{BBCC}  is \emph{not} the original CEO, or Bouchard-Eynard topological recursion. It is rather called the topological recursion \emph{without branch covers} which was first presented in \cite[Section 10]{ABCD} for $r=2$. Its dual description in terms of Airy structures was first realised in \cite[Section 4.1.3]{SAS} for $r=2$, and BBCC significantly generalised for all $r\geq2$ in \cite{BBCC}. This difference originates from the fact that the Gaiotto vector is constructed in an untwisted module of the $\mathcal{W}(\mathfrak{gl}_r)$-algebra whereas the CEO or Bouchard-Eynard topological recursion is related to twisted modules. Below is a schematic summary of relations between the topological recursion and Airy structures for the $\mathcal{W}(\mathfrak{gl}_r)$-algebra:

\begin{figure}[h]
\begin{tikzpicture}[every text node part/.style={align=center}]
\node (c) {};
\node[entity, below right=0mm and 5mm of c] (TR1) {Topological recursion on  \\ an $r$-ramified spectral curve};
\node[entity, below left=0mm and 5mm of c] (AS1) {Airy structure as a \emph{twisted} \\ module of the $\mathcal{W}(\mathfrak{gl}_r)$-algebra};
\draw[<->, shorten >= 2pt, shorten <= 2pt, draw=black,thick] (TR1) -- node [text width=2cm,midway,above ] {1\;:\;1}  (AS1);
\node[entity, below right=18mm and 5mm of c] (TR2) {Topological recursion \\ without branch covers};
\node[entity, below left=18mm and 5mm of c] (AS2) {Airy structure as an \emph{untwisted} \\ module of the $\mathcal{W}(\mathfrak{gl}_r)$-algebra};
\draw[<->, shorten >= 2pt, shorten <= 2pt, draw=black,thick] (TR2) -- node [text width=2cm,midway,above ] {1\;:\;1}  (AS2);
\end{tikzpicture}
\caption{Correspondences between Airy structures and the topological recursion. The first one is for the topological recursion of CEO \cite{EO2} and Bouchard-Eynard \cite{BE}. The work of BBCC \cite{BBCC} is the second one. 
}\label{fig:1}
\end{figure}
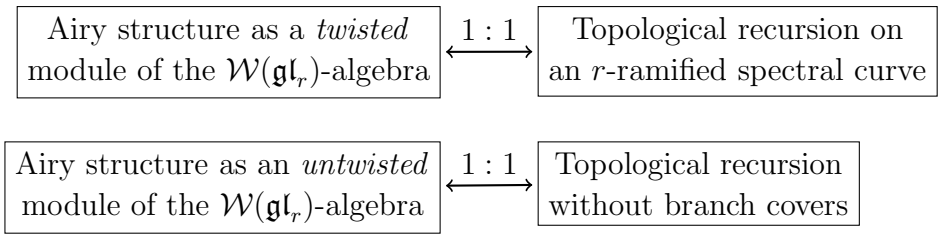

\subsection{Goal of This Paper}
A supersymmetric generalisation of Airy structures, called ``super Airy structures'', was introduced by the author and five more collaborators in \cite{SAS}. Following it, the ``$\mathcal{N}=1$ super topological recursion'' -- a supersymmetric generalisation of the CEO topological recursion -- was recently proposed by Bouchard and the author in \cite{BO2}. Moreover,  \cite{BO2} shows an equivalence between the $\mathcal{N}=1$ super topological recursion on a local super spectral curve and the corresponding super Airy structure as a twisted module of $\mathcal{N}=1$ super Virasoro algebra. This means that a supersymmetric generalisation of the first correspondence in Figure~\ref{fig:1} for $r=2$ has been established by \cite{SAS,BO2}.

Now, motivated by the work of BBCC \cite{BBCC}, one may ask: for $r=2$, can the second correspondence in Figure~\ref{fig:1} be also generalised with $2d$ supersymmetry? And more interestingly, can we incorporate $2d$ supersymmetry into an interesting relation \eqref{BBCC}? The goal of the present paper is to show that the answers are, YES, to both questions. Here we give an outline of the conceptual path to achieve this goal.

 \cite{SAS} presented four different classes of super Airy structures as modules of the $\mathcal{N}=1$ super Virasoro algebra (see Table~\ref{table:1}). The untwisted and the $\mu$-twisted module play a role of building blocks in the present paper, and we suitably generalise discussions in \cite{SAS} to match with the corresponding super topological recursion. We note that the super Airy structures discussed in \cite{BO2} reduce down to the $\rho$-twisted module, which is indeed a natural supersymmetric extension of the CEO topological recursion. 
 
 \begin{table}[h]
 \centering
\begin{tabular}{ | c || c | c | c | } 
\hline
classes & boson & fermion & sector \\ 
\hline
untwisted & $\times$ & $\times$ & NS \\ 
\hline
$\mu$-twisted & $\times$ & $\Circle$ & R \\ 
\hline
$\sigma$-twisted & $\Circle$ & $\times$ & R \\ 
\hline
$\rho$-twisted & $\Circle$ & $\Circle$ & NS \\ 
\hline
\end{tabular}
\vspace{3mm}
\caption{Four classes of super Airy structures as modules of the $\mathcal{N}=1$ super Virasoro algebra \cite{SAS}. $\Circle$ and $\times$ denotes twisted and untwisted respectively, and the resulting super Virasoro algebra is either in the Neveu-Schwarz (NS) sector or Ramond (R) sector.}\label{table:1}
\end{table}

On the other hand, untwisted/$\mu$-twisted super spectral curves and the corresponding super topological recursion are new concepts, hence we will define them in the present paper. These super spectral curves would be supersymmetric analogues of spectral curves \emph{without branch covers}. However, since we do not have a good understanding of global super spectral curves at the moment, we avoid using the terminology ``without branch covers'', but rather adapt the notation of twisting from \cite{SAS}. Note that in our notation, the definition of a local super spectral curve of \cite[Section 2.2]{BO2} would be named as a $\rho$-twisted super spectral curve. We then define untwisted and $\mu$-twisted abstract super loop equations, which play a key role to show an equivalence between the super topological recursion and super Airy structures. These relations are schematically summarised as follows:

\begin{figure}[h]
\centering
\begin{tikzpicture}
\node[entity, align=center] (eq) {untwisted or $\mu$-twisted\\ abstract super loop equations};
\node[below left= 7mm and -2cm of eq, align=left] (P) {Solve geometrically \\ -- residue analysis};
\node[below right= 7mm and -2cm of eq, align=left] (V) {Solve algebraically \\ -- super Virasoro constraints};
\node[entity, align=center, below=7mm of P] (TR) {untwisted or $\mu$-twisted\\ super topological recursion};
\node[entity, align=center, below=7mm of V] (AS) {super Airy structure as \\ untwisted or $\mu$-twisted modules};
\draw[-,draw=black, thick] (eq) to (P);
\draw[-,draw=black, thick] (eq) to (V);
\draw[->,draw=black,thick] (P) to (TR);
\draw[->,draw=black,thick] (V) to (AS);
\end{tikzpicture}
\caption{One of the goals of the present paper is to mathematically formulate the above flowchart. We note that the above flowchart with $\rho$-twisting is none other than the work of \cite{BO2}.}\label{fig:goal}
\end{figure}
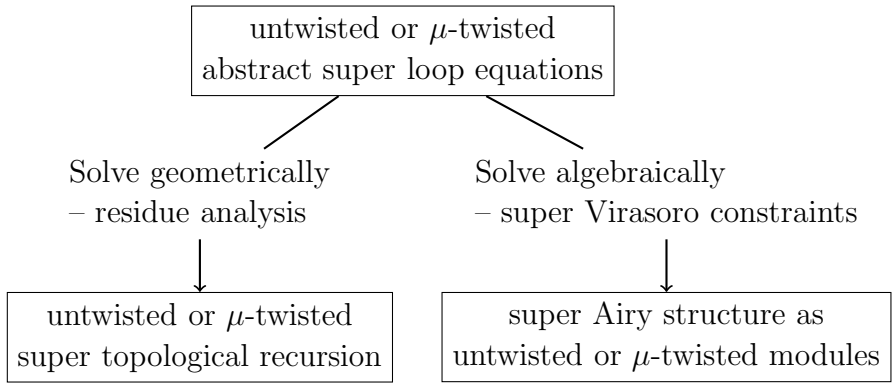

Finally, it was conjectured by \cite{BBM,Ito} that the norm of the Gaiotto vector for $\mathcal{N}=1$ superconformal blocks corresponds to the Nekrasov partition function of gauge group $U(2)$ on $\mathbb{C}^2/\mathbb{Z}_2$. We will explore a relation between the Gaiotto vector in the Neveu-Schwarz/Ramond sector and the partition function of a super Airy structure as an untwisted/$\mu$-twisted module of the $\mathcal{N}=1$ super Virasoro algebra. In particular, we prove the existence of the Gaiotto vectors for superconformal blocks which is assumed in physics literature. From this perspective, the framework of super Airy structures is a bridge connecting the super topological recursion and Gaiotto vectors for superconformal blocks. This stands as a supersymmetric generalisation of \eqref{BBCC}.

\subsection{Organisation}
This paper is organised as follows. In Section~\ref{sec:STR}, we will introduce the notion of an untwisted/$\mu$-twisted super spectral curve as well as abstract super loop equations. We then present an associated recursive formalism that solves the abstract super loop equations, which we call the untwisted/$\mu$-twisted super topological recursion. In Section~\ref{sec:SAS}, we construct the corresponding super Airy structure for each untwisted/$\mu$-twisted super spectral curve, and show a one-to-one relation to the super topological recursion. This equivalence is stated in Theorem~\ref{thm:main} which is the main theorem of the present paper. In Section~\ref{sec:SCB}, we recall basic facts about the Gaiotto vectors for $\mathcal{N}=1$ super conformal blocks, and review a conjectural relation between the Gaiotto vectors and the Nekrasov partition function of pure $U(2)$ theory on $\mathbb{C}^2/\mathbb{Z}_2$. We then prove in Proposition~\ref{prop:main} that the Gaiotto vector in the Neveu-Schwarz or Ramond sector is none other than the partition function of an appropriate super Airy structure,  with a simple change of parameters. We conclude in Section~\ref{sec:Conclusion} with open problems and possible future directions. Details of computations are in part given in Appendix~\ref{sec:App}

\subsubsection*{Acknowledgements 
}The author owes many thanks to Nitin Chidambaram for inspirational discussions and for helpful explanations of the recent results of BBCC \cite{BBCC}.
The author thanks Andrea Brini, Omar Kidwai, and Piotr Su\l kowski for various comments.
The author also acknowledges the Young Researchers Integrability School and Workshop 2020 (YRISW) for a series of introductory lectures at which the author learnt about some topics discussed in the present paper. This work is supported by the Engineering and Physical Sciences Research Council under grant agreement ref. EP/S003657/2, by the TEAM programme of the Foundation for Polish Science co-financed by the European Union under the European Regional Development Fund (POIR.04.04.00-00-5C55/17- 00), and also during the revision by Grant-in-Aid for JSPS Fellows (KAKENHI Grant Number: 22J00102). Finally, the author appreciates various helpful comments by the editors and referees for the improvements of the presentation.

\section{Super Topological Recursion}\label{sec:STR}
In this section, we develop a supersymmetric analogue of the topological recursion without branch covers. We introduce the definitions of an untwisted/$\mu$-twisted super spectral curve and abstract super loop equations, and construct the corresponding recursive formalism. We note that they differ from the definitions of a local super spectral curve and the super topological recursion recently proposed in \cite{BO2} which would be called a $\rho$-twisted super spectral curve and the $\rho$-twisted super topological recursion in our notation --- a natural supersymmetric generalisation of the CEO topological recursion. Parts of the presentation follow \cite{BO2}.

\subsection{Untwisted/$\mu$-twisted Super Spectral Curves}
There are three symplectic vector spaces underlying untwisted/$\mu$-twisted super spectral curves, namely, the vector space for
\begin{enumerate}
    \item bosons $V_{z}^B$,
    \item untwisted fermions $V_{z,\theta}^{NS}$,
    \item $\mu$-twisted fermions $V_{z,\theta}^R$.
\end{enumerate}
Since untwisted/$\mu$-twisted fermions correspond to fermions in the Neveu-Schwarz/Ramond sector respectively, we often denote them by NS/R instead. 

We start with the vector space $V_{z}^B$ for bosons defined in the previous section:
\subsubsection{Vector Space for Bosons}
Let us define $V_{z}^B$ by
\begin{equation}
V_{z}^B:=\{\;\omega\in\mathbb{C}[z^{-1},z\rrbracket dz\;\;|\;\;\underset{z\rightarrow0}{\text{Res}}\,\omega(z)=0\;\},
\end{equation}
equipped with the following symplectic pairing $\Omega^B:V_z^B\times V_z^B\rightarrow\mathbb{C}$:
\begin{equation}
df_1,df_2\in V_z^B,\;\;\;\;\Omega^B(df_1,df_2)=\underset{z\rightarrow0}{\text{Res}} f_1(z)df_2(z).\label{inner product}
\end{equation}
Note that $\Omega^B$ makes sense because no vector in $V_{z}^B$  has residues. We then define a Lagrangian subspace $V_{z}^{B+}=\mathbb{C}\llbracket  z\rrbracket dz\subset V_{z}^B$, with a choice of a basis  $(d\xi_l)_{l>0}$ as
\begin{equation}
d\xi_l(z):=z^{l-1}dz,\;\;\;\;\;l\in\mathbb{Z}_{>0}.
\end{equation}

We consider another Lagrangian subspace $V_{z}^{B-}\subset V_{z}^B$ complementary to $V_z^{B+}$, but this is not unique, and there is a choice of so-called ``bosonic polarization parameters'' $\phi_{lm}=\phi_{ml}\in\mathbb{C}$ for all $k,l\in\mathbb{Z}_{>0}$.  If we denote by $(d\xi_{-l})_{l>0}$ a basis of $V_z^{B-}$, bosonic polarization parameters $\phi_{lm}$ appear in $(d\xi_{-l})_{l>0}$ as:
\begin{equation}
d\xi_{-l}(z)=\frac{dz}{z^{l+1}}+\sum_{m>0}\frac{\phi_{lm}}{l}d\xi_m(z),\;\;\;\;l\in\mathbb{Z}_{>0}.\label{bbasis}
\end{equation}
One can easily check
\begin{equation}
\forall k,l\in\mathbb{Z}_{\neq0},\;\;\;\;\Omega^B(d\xi_k,d\xi_l)=\frac{\delta_{k+l,0}}{k},\label{vecHeis}
\end{equation}
hence $V_z^{B-}$ is indeed complementary to $V_z^{B+}$. One may find \eqref{vecHeis} similar to commutation relations in the Heisenberg algebra.

The choice of polarization can be encoded into a formal symmetry bidifferential $\omega_{0,2|0}$ as follows,
\begin{equation}
\omega_{0,2|0}(z_1,z_2|)=\frac{dz_1 dz_2}{(z_1-z_2)^2}+\sum_{k,l>0}\phi_{kl}\;d\xi_k(z_1) d\xi_l(z_2).
\end{equation}
In the domain $|z_1|>|z_2|$, the basis of $V_z^{B\pm}$ appears in the expansion of $\omega_{0,2|0}(z_1,z_2|)$ as
\begin{equation}
\omega_{0,2|0}(z_1,z_2|)=\sum_{l\geq1}ld\xi_{-l}(z_1)d\xi_l(z_2).\label{02z=0}
\end{equation}

In order to define a spectral curve without branch covers, we further consider a one-dimensional vector space $V_z^{B\,0}$ and chose its basis $d\xi_0$ by
\begin{equation}
d\xi_0(z)=\frac{dz}{z}.
\end{equation}
Note that $d\xi_0\not\in V_z^B$. It does not make sense to apply the definition of the symplectic form \eqref{inner product} to $\tilde{V}_z^B=V_z^B\oplus V^{B\,0}_z$ because $\int d\xi_0=\log z$. We can still formally define a degenerate symplectic form $\tilde \Omega^B:\tilde{V}_z^B\times\tilde{V}_z^B\to\mathbb{C}$ such that if $k,l\neq0$ then $\tilde \Omega^B(d\xi_k,d\xi_l)$ follows \eqref{inner product} and $\tilde \Omega^B(d\xi_0,d\xi_k)=0$ for all $k\in\mathbb{Z}$. Thus, one can interpret that $d\xi_0\in V_z^{B\,0}$ plays the same role as the Heisenberg zero mode.

We now define a spectral curve without branch covers:

\begin{definition}\label{def:curveB}
	A  \emph{spectral curve with one component without branch covers} consists of a degenerate symplectic vector space $\tilde{V}_z^B$ with a maximal isotropic subspace $V_z^{B+}$, and the following data:
\begin{itemize}
\item a choice of ``dilaton shift parameters'' $(\tau_l)_{l\geq-(N-1)}$ with $\tau_{N-1}\neq0$ and $N\in\mathbb{Z}_{>0}$ which can be encoded in a choice of a one-form $\omega_{0,1|0}\in \tilde V^{B}_z$:
\begin{equation}
\omega_{0,1|0}(z)=\sum_{l\geq-(N-1)}\tau_{l}d\xi_{l}(z),
	\end{equation}
\item a choice of ``bosonic polarization parameters'', which can be encoded in a choice of a symmetric bidifferential $\omega_{0,2|0}$:
\begin{equation}
\omega_{0,2|0}(z_1,z_2|)=\frac{dz_1 dz_2}{(z_1-z_2)^2}+\sum_{k,l>0}\phi_{kl}\;d\xi_k(z_1) d\xi_l(z_2).
\end{equation}
\item a choice of ``$D$-terms'' $D_k$ for $1\leq k\leq N$ which can be encoded in a choice of a one-form $\omega_{1,1}\in \tilde V^{B-}_z$:
\begin{equation}
\omega_{1,1|0}(z|)=\sum_{l=1}^N D_k\xi_{-k}.
\end{equation}
\item a choice of ``crosscap parameters'' $(Q_l)_{l\geq-(N-1)}$ with $N\in\mathbb{Z}_{>0}$ which can be encoded in a choice of a one-form $\omega_{\frac12,1|0}\in V^{B}_z$:
\begin{equation}
\omega_{\frac12,1|0}(z|)=\sum_{l\geq-(N-1)}Q_{l}d\xi_{l}(z),
\end{equation}
\end{itemize}
\end{definition}

\begin{remark}
The above definition is similar to those in \cite{ABCD,BBCC} but slightly different. BBCC \cite{BBCC} considers a spectral curve with $(r-1)$ components but only for $N=1$. We show that duality between the (super) topological recursion and (super) Airy structures holds for arbitrary positive integers $N$, though applications to Gaiotto vectors require to set $N=1$.
\end{remark}

\subsubsection{Vector Space for Untwisted Fermions}

We move on to the definition of the vector space $V_{z,\theta}^{NS}$ for untwisted fermions. Let us first summarise important properties of the vector space for bosons:
\begin{itemize}
    \item \;\;$V_z^B$ \,: space of meromophic differentials on $\mathbb{C}_z$,
    \item \;\;$\Omega^B$ \,: anti-symmetric form defined by taking a residue at $z=0$,
    \item $\omega_{0,2|0}$ : symmetric bidifferential which encodes the information about the Lagrangian subspace $V_z^{B-}$ (bosonic polarization parameters $\phi_{kl}$).
\end{itemize}
Recall that spin-statistics is swapped when one replaces fermions in place of bosons. Thus, it is expected from a physical point of view that the vector space $V_{z,\theta}^{NS}$ for untwisted fermions consists of:
\begin{itemize}
    \item \;\;$V_{z,\theta}^{NS}$ \,: space of meromophic \emph{half-order} differentials on $\mathbb{C}^{1|1}_{z|\theta}$,
    \item \;\;$\Omega^{NS}$ \,: \emph{symmetric} form defined by taking a residue at $z=0$,
    \item \;$\omega_{0,2|0}$\; : \emph{anti-symmetric} half-order bidifferential which encodes the information about the Lagrangian subspace $V_{z,\theta}^{NS-}$.
\end{itemize}
For untwisted fermions, i.e., fermions in the Neveu-Schwarz sector, it turns out that such a naive exchange of spin-statistics works for untwisted fermions, and we will construct a vector space that admits these properties.

The vector space $V^{NS}_{z,\theta}$ for untwisted fermions is 
\begin{equation}
	V^{NS}_{z,\theta}:=\{\eta\in\mathbb{C}[z^{-1},z\rrbracket\;\Theta^{NS}(z,\theta)\},
\end{equation}
where
\begin{equation}
	\Theta^{NS}(z,\theta):=\left(\theta+dz\frac{\partial}{\partial\theta}\right),
\end{equation}
and $\theta$ is a Grassmann variable. We equip $V^{NS}_{z,\theta}$ with a pairing $\Omega^F:V^F_{z,\theta}\times V^F_{z,\theta}\rightarrow\mathbb{C}$
\begin{equation}
\Omega^{NS}(\eta_1,\eta_2):=\underset{z\rightarrow0}{\text{Res}}\;\eta_1(z,\theta)\eta_2(z,\theta),\label{NS product}
\end{equation}
We often denote  $\Theta^{NS}(z,\theta)$ and $\Theta^{NS}(z_i,\theta_i)$ by  $\Theta^{NS}_z$ and $\Theta^{NS}_i$ for brevity. We also omit the $\theta$-dependence below. Note that $(\Theta_z^{NS})^2 = dz$, hence, $\Theta^{NS}_z$ can be thought of as a half-order differential, or as a nonzero section of a completely integrable subbundle $\tilde{\mathcal{D}}\subset T^*\mathbb{C}_z$ of rank $0|1$\footnote{In the context of superconformal structures of super Riemann surfaces, a completely nonintegrable subbundle $\mathcal{D}$ of rank $0|1$ normally refers to $\mathcal{D}\subset T\mathbb{C}_z$ (e.g., \cite[Section 2.2]{W}) instead of $T^*\mathbb{C}_z$. Thus, one might find $\tilde{\mathcal{D}}$ a dual description of $\mathcal{D}$.}. 

We decompose $V_z^{NS}$ into two Lagrangian subspaces $V_z^{NS+}$ and $V_z^{NS-}$ as follows. We define $V_z^{NS+}=\{\eta\in \mathbb{C}\llbracket z\rrbracket\,\Theta_z^{NS}\}$, and we fix its basis $(\eta_{l+\frac12})_{l\geq0}$ by
\begin{equation}
\eta_{l+\frac12}(z,\theta):=z^{l}\,\Theta_z,\;\;\;\;\;l\in\mathbb{Z}_{\geq0}.
\end{equation}
Let $V_z^{NS-}$ be another Lagrangian subspace complementary to $V_z^{NS+}$, then its basis $(\eta_{-l-\frac 12})_{l\geq0}$ is given with ``untwisted polarization parameters'' $\psi^{NS}_{lm}=-\psi^{NS}_{ml}\in\mathbb{C}$ as
\begin{equation}
\eta_{-l-\frac 12}(z,\theta):=\left(\frac{1}{z^{l+1}}+\sum_{k\geq0}\psi^{NS}_{lk}z^{k}\right)\,\Theta,\;\;\;\;l\in\mathbb{Z}_{\geq0}.\label{nsbasis}
\end{equation}
Note that we have
\begin{equation}
\forall k,l\in\mathbb{Z},\;\;\;\;\Omega^{NS}(\eta_{k+\frac12},\eta_{l-\frac12})=\delta_{k+l,0},
\end{equation}
which resembles commutation  relations in the Clifford algebra. Analogous to $\omega_{0,2|0}$, we encode untwisted polarization parameters into an antisymmetric bidifferential  $\omega^{NS}_{0,0|2}$ which is defined as
\begin{equation}
\omega^{NS}_{0,0|2}(|z_1,z_2)=-\frac{\Theta^{NS}_1\Theta^{NS}_2}{(z_1-z_2)}+\sum_{k,l\geq0}\psi^{NS}_{kl}\;d\eta_{k+\frac 12}(z_1) d\eta_{l+\frac12}(z_2).
\end{equation}
Then, in the domain $|z_1|<|z_2|$, the basis of $V_z^{NS\pm}$ appears in the expansion of $\omega_{0,0|2}^{NS}(|z_1,z_2)$ as
\begin{equation}
\omega^{NS}_{0,0|2}(|z_1,z_2)=\sum_{l\geq0}\eta_{l+\frac12}(z_1)\eta_{-l-\frac12}(z_2).\label{NS02expansion}
\end{equation}

We now have all ingredients to define an untwisted super spectral curve:

\begin{definition}\label{def:curveNS}
	An  \emph{untwisted super spectral curve $\mathcal{S}_{NS}$ with one component} consists of a $\mathbb{Z}_2$-graded degenerate symplectic vector space $V=\tilde{V}_z^B\oplus V^{NS}_z$ with a maximal isotropic  subspace $V_z^{B+}\oplus V_z^{NS+}$, and the following data:
\begin{itemize}
\item a choice of ``dilaton shift parameters'' $(\tau_l)_{l\geq-(N-1)}$ with $\tau_{N-1}\neq0$ and $N\in\mathbb{Z}_{>0}$ which can be encoded in a choice of a one-form $\omega_{0,1|0}\in \tilde V^{B}_z$:
\begin{equation}
\omega_{0,1|0}(z)=\sum_{l\geq-(N-1)}\tau_{l}d\xi_{l}(z),
	\end{equation}
\item a choice of ``bosonic polarization parameters'', which can be encoded in a choice of a symmetric bidifferential $\omega_{0,2|0}$:
\begin{equation}
\omega_{0,2|0}(z_1,z_2|)=\frac{dz_1 dz_2}{(z_1-z_2)^2}+\sum_{k,l>0}\phi_{kl}\;d\xi_k(z_1) d\xi_l(z_2).
\end{equation}
\item a choice of ``$D$-terms'' $D_k$ for $1\leq k\leq N$ which can be encoded in a choice of a one-form $\omega_{1,1}\in V^{B-}_z$:
\begin{equation}
\omega_{1,1|0}(z|)=\sum_{l=1}^N D_k\xi_{-k}.
\end{equation}
\item a choice of ``crosscap parameters'' $(Q_l)_{l\geq-(N-1)}$ with $N\in\mathbb{Z}_{>0}$ which can be encoded in a choice of a one-form $\omega_{\frac12,1|0}\in \tilde V^{B}_z$:
\begin{equation}
\omega_{\frac12,1|0}(z|)=\sum_{l\geq-(N-1)}Q_{l}d\xi_{l}(z),
\end{equation}
\item a choice of ``untwisted polarization parameters'', which can be encoded in a choice of an antisymmetric bidifferential $\omega^{NS}_{0,2|0}$:
\begin{equation}
\omega^{NS}_{0,0|2}(|z_1,z_2)=-\frac{\Theta^{NS}_1\Theta^{NS}_2}{(z_1-z_2)}+\sum_{k,l\geq0}\psi^{NS}_{kl}\;d\eta_{k+\frac 12}(z_1) d\eta_{l+\frac12}(z_2).
\end{equation}
\end{itemize}
\end{definition}

\subsubsection{Vector Spaces for $\mu$-twisted Fermions}
At last we construct the vector space $V^{R}_{z,\theta}$ for $\mu$-twisted fermions, which already appeared in \cite{BO2}. Although some aspects of $V^{R}_{z,\theta}$ can be similarly defined as we did for $V^{NS}_{z,\theta}$, a critical difference is that there exits the zero mode in the Ramond sector. It turns out that existence of the zero mode results in delicate geometric treatments unlike the naive construction of $V_{z,\theta}^{NS}$ by exchanging spin-statistics from $V_z^B$.

For $\mu$-twisted fermions, the vector space $V^{R}_{z,\theta}$ is given by
\begin{equation}
	V^{R}_{z,\theta}:=\{\eta\in\mathbb{C}[z^{-1},z\rrbracket\;\Theta^R(z,\theta)\},
\end{equation}
where
\begin{equation}
	\Theta^R(z,\theta):=\left(\theta+zdz\frac{\partial}{\partial\theta}\right),
\end{equation}
and $\theta$ is a Grassmann variable. We equip $V^F$ with a pairing $\Omega^R:V^R_{z,\theta}\times V^R_{z,\theta}\rightarrow\mathbb{C}$
\begin{equation}
\Omega^R(\eta_1,\eta_2):=\underset{z\rightarrow0}{\text{Res}}\;\eta_1(z,\theta)\eta_2(z,\theta),\label{R product}
\end{equation}
We again denote  $\Theta^R(z,\theta)$ and $\Theta^R(z_i,\theta_i)$ by  $\Theta^{R}_z$ and $\Theta^{R}_i$ for brevity. Note that $(\Theta_z^R)^2 = z dz$, hence, $\Theta_z^R$ fails to be a nonzero section of the completely integrable subbundle $\tilde{\mathcal{D}}\subset T^*\mathbb{C}_z$ of rank $0|1$. Thus, one may interpret that $V^{R}_{z}$ is associated with a Ramond divisor at the origin in the context of \cite[Section 4.1]{W}. 

As explained in \cite[Section 2.2]{BO2}, $V^{R}_{z}$ are decomposed into three subspaces $V_z^{R+},V_z^{R\,0},V_z^{R-}$. Similar to $V_z^{B+}, V_z^{NS+}$, we define $V_z^{R+}=\{\eta\in \mathbb{C}\llbracket z\rrbracket\,\Theta_z^R\}$, and we fix its basis $(\eta_l)_{l>0}$ as
\begin{equation}
\eta_l(z,\theta):=z^{l-1}\,\Theta^R_z,\;\;\;\;\;l\in\mathbb{Z}_{>0}.
\end{equation}
Next, the zero mode space $ V^{R\,0}$ is one-dimensional whose basis $(\eta_0)$ is given by
\begin{equation}
\eta_0(z,\theta):=\left(\frac{1}{z}+\sum_{k>0}\psi^R_{0k}z^{k-1}\right)\,\Theta_z^R,\label{R basis0}
\end{equation}
where $\psi_{0k}\in\mathbb{C}$. At last, $V_z^{R-}$ is complementary to $V_z^{R+}\oplus V_z^{R\,0}$ whose basis $(\eta_{-l})_{l>0}$ is given by
\begin{equation}
\eta_{-l}(z,\theta):=\left(\frac{1}{z^{l+1}}+\sum_{k\geq0}\psi^R_{lk}z^{k-1}\right)\,\Theta_z^R,\label{R basis}
\end{equation}
where ``$\mu$-twisted polarization parameters'' $\psi^R_{kl}$  satisfy
\begin{equation}
\psi_{00}^R=0,\;\;\;\;\psi^R_{kl}+\psi^R_{lk}+\psi^R_{0k}\psi^R_{0l}=0,\;\;\;\;\forall k,l\in\mathbb{Z}_{\geq0}.\label{R polarization parameters}
\end{equation}
One can easily check that $(\eta_{l})_{l\in\mathbb{Z}}$ satisfy
\begin{equation}
\forall k,l\in\mathbb{Z},\;\;\;\;\Omega^R(\eta_k,\eta_l)=\delta_{k+l,0},
\end{equation}
which again resembles commutation relations in the Clifford algebra. In order to put all $\mu$-twisted polarization parameters into one package, we define an antisymmetric bidifferential $\omega_{0,0|2}^R$ by
\begin{equation}
\omega_{0,0|2}^R(|z_1,z_2):=-\frac12\frac{z_1+z_2}{z_1-z_2}\frac{\Theta^R_1 \Theta^R_2}{z_1z_2}-\sum_{k,l\geq1}\frac{\psi^R_{k-1\;l-1}-\psi^R_{l-1\;k-1}}{1+\delta_{(k-1)(l-1),0}}\frac{\eta_l(z_1)  \eta_k(z_2)}{2z_1z_2}.
\end{equation}
Then, in the domain $|z_1|<|z_2|$, we have
\begin{equation}
\omega^R_{0,0|2}(|z_1,z_2)\rightarrow\sum_{l>0}\eta_{l}(z_1)\eta_{-l}(z_2)+\frac12\eta_0(z_1)\eta_0(z_2).\label{R02expansion}
\end{equation}

A $\mu$-twisted super spectral curve is defined analogously to an untwisted one:

\begin{definition}\label{def:curveR}
	A  \emph{$\mu$-twisted super spectral curve $\mathcal{S}_R$  with one component} consists of a $\mathbb{Z}_2$-graded degenerate symplectic vector space $V=\tilde{V}_z^B\oplus V^{R}_z$, with a maximal isotropic subspace $V_z^{B+}\oplus V_z^{R+}$, and the following data:
\begin{itemize}
\item a choice of ``dilaton shift parameters'' $(\tau_l)_{l\geq-(N-1)}$ with $\tau_{N-1}\neq0$ and $N\in\mathbb{Z}_{>0}$ which can be encoded in a choice of a one-form $\omega_{0,1|0}\in \tilde V^{B}_z$:
\begin{equation}
\omega_{0,1|0}(z)=\sum_{l\geq-(N-1)}\tau_{l}d\xi_{l}(z),
	\end{equation}
\item a choice of ``bosonic polarization parameters'', which can be encoded in a choice of a symmetric bidifferential $\omega_{0,2|0}$:
\begin{equation}
\omega_{0,2|0}(z_1,z_2|)=\frac{dz_1 dz_2}{(z_1-z_2)^2}+\sum_{k,l>0}\phi_{kl}\;d\xi_k(z_1) d\xi_l(z_2).
\end{equation}
\item a choice of ``$D$-terms'' $D_k$ for $1\leq k\leq N$ which can be encoded in a choice of a one-form $\omega_{1,1}\in V^{B-}_z$:
\begin{equation}
\omega_{1,1|0}(z|)=\sum_{l=1}^N D_k\xi_{-k}.
\end{equation}
\item a choice of ``crosscap parameters'' $(Q_l)_{l\geq-(N-1)}$ with $N\in\mathbb{Z}_{>0}$ which can be encoded in a choice of a one-form $\omega_{\frac12,1|0}\in \tilde V^{B}_z$:
\begin{equation}
\omega_{\frac12,1|0}(z|)=\sum_{l\geq-(N-1)}Q_{l}d\xi_{l}(z),
\end{equation}
\item a choice of ``$\mu$-twisted polarization parameters'', which can be encoded in a choice of an antisymmetric bidifferential $\omega^R_{0,2|0}$:
\begin{equation}
\omega^{R}_{0,0|2}(|z_1,z_2)=-\frac12\frac{z_1+z_2}{z_1-z_2}\frac{\Theta^R_1 \Theta^R_2}{z_1z_2}-\sum_{k,l\geq1}\frac{\psi^R_{k-1\;l-1}-\psi^R_{l-1\;k-1}}{1+\delta_{(k-1)(l-1),0}}\frac{\eta_l(z_1)  \eta_k(z_2)}{2z_1z_2}.
\end{equation}
\end{itemize}
\end{definition}

\subsection{Untwisted and $\mu$-twisted Abstract Super Loop Equations}
It turns out that the forms of untwisted/$\mu$-twisted abstract super loop equations and the corresponding super topological recursion are almost universal regardless of the twisting. In addition, all super spectral curves considered in the present paper are with one component. Thus, we often drop ``untwisted/$\mu$-twisted ... with one component'', and we simply call $\mathcal{S}_F$ a super spectral curve where $F\in\{NS,R\}$. For brevity of presentations, we further define a constant $f$ such that:
\begin{align}
\text{for }F=NS,\;\;\;\;f=\frac12,\;\;\;\;
\text{for }F=R,\;\;\;\;\;\;\;f=0\label{fandNSR}
\end{align}

\begin{remark}\label{rem:para}
A super spectral curve $\mathcal{S}_F$ is completely fixed by the following seven parameters:
\begin{equation}
(F,N,\tau_{l},\phi_{kl},\psi_{mn}^F,D_k,Q_l)
\end{equation}
\end{remark}

\subsubsection{Abstract Super Loop Equations}
Given a super spectral curve $\mathcal{S}_F$, we consider an infinite sequence of multilinear differentials $\omega_{g,n|2m}$ for $2g,n,m\in\mathbb{Z}_{\geq0}$ with $2g+n+2m\geq3$ as
\begin{equation}
	\omega_{g,n|2m}\in \left(\bigotimes_{j=1}^nV_{z_j}^{B-}\right)\otimes \left(\bigotimes_{k=1}^{2m} V_{u_k,\theta_k}^{F\,0,-} \right),
\end{equation}
where we should drop $V_{u_k,\theta_k}^{F\,0}$ when $F=NS$. Note that $g$ can be a \emph{half-integer}. We say that $\omega_{g,n|2m}$ ``respect the polarization'', because they are in the subspaces $V_{z_j}^{B-}$ and $V_{u_k,\theta_k}^{F\, 0,-}$ defined by the choice of polarization in the super spectral curve $\mathcal{S}_F$. We impose that the $\omega_{g,n|2m}$ are symmetric under permutations of the first $n$ entries, and anti-symmetric under permutations of the last $2m$ entries. We assume no symmetry under permutations of some of the first $n$ entries with some of the last $2m$ entries. Note that the $\omega_{g,n|2m}$ always have an even number of elements in $\bigotimes V_{u,\theta}^{F\,0,-}$. 

Let us denote by $I,J$ a set of variables $I=(z_1,...)$ and $J=((u_1,\theta_1),..)$. The number of variables in $I$ and $J$ should be read off from multilinear differentials of interests, e.g., for $\omega_{g,n+1|2m}(z_0,I|(u_0,\theta_0),J)$ we have  $|I|=n$ and $|J|=2m-1$. For $2g+n+2m\geq3$ we define:
\begin{align}
\mathcal{Q}_{g,n|2m}^{BF}(I|z,J)=\;&\omega_{g-1,n+1|2m}(z,I|z,J)\nonumber\\
&+\left(\underset{\tilde z\rightarrow0}{{\rm Res}}\,\omega_{\frac12,1|0}(\tilde z)\right)\left(\mathcal{D}_z\cdot\omega_{g-\frac12,n|2m}(I|z,J)+\frac{1-2f}{2}d\xi_0(z)\omega_{g-\frac12,n|2m}(I|z,J)\right)\nonumber\\
&+\sum_{g_1+g_2=g}\sum_{\substack{I_1\cup I_2=I \\ J_1\cup J_2=J}}(-1)^{\rho}\omega_{g_1,n_1+1|2m_1}(z,I_1|J_1)\,\omega_{g_2,n_2|2m_2}(I_2|z,J_2),\label{QFB}
\end{align}
where we dropped the $\theta$-dependence for brevity. $(-1)^{\rho}=1$ if $J_1\cup J_2$ is an even permutation of $J$ and $(-1)^{\rho}=-1$ otherwise. Similarly, for  $2g+n+2m\geq2$ and $(g,n,m)\neq(1,0,0)$, we define
\begin{align}
\mathcal{Q}_{g,n+1|2m}^{BB}(z,I|J)=\;&\omega_{g-1,n+2|2m}(z,z,I|J)+\frac12\left(\underset{\tilde z\rightarrow0}{{\rm Res}}\,\omega_{\frac12,1|0}(\tilde z)\right)\mathcal{D}_z\cdot\omega_{g-\frac12,n+1|2m}(z,I|J)\nonumber\\
&+\sum_{g_1+g_2=g}\sum_{\substack{I_1\cup I_2=I \\ J_1\cup J_2=J}}(-1)^{\rho}\omega_{g_1,n_1+1|2m_1}(z,I_1|J_1)\,\omega_{g_2,n_2+1|2m_2}(z,I_2|J_2),\label{QBB}
\end{align}
\begin{align}
\mathcal{Q}_{g,n+1|2m}^{FF}(z,I|J)=\;&-\frac12\mathcal{D}_z\cdot\omega_{g-1,n|2m+2}(I|z,u,J)\Bigr|_{u=z}\nonumber\\
&+\frac12\sum_{g_1+g_2=g}\sum_{\substack{I_1\cup I_2=I \\ J_1\cup J_2=J}}(-1)^{\rho}\mathcal{D}_z\cdot\omega_{g_1,n_1|2m_1}(I_1|z,J_1) \,\omega_{g_2,n_2|2m_2}(I_2|z,J_2),\label{QFF}
\end{align}
where for $\omega(z)=f(z)dz\in V_{z}^B$ and  $\eta(z,\theta)=g(z)\Theta^F(z,\theta)\in V_{z,\theta}^F$, the derivative operator $\mathcal{D}_z$ is defined to act as
\begin{align}
\mathcal{D}_z\cdot\omega(z)&=df(z) dz\in V_z^B\otimes V_z^B,\\
\mathcal{D}_z\cdot\eta(z,\theta)&=dg(z) \Theta^F(z,\theta)\in V_z^B\otimes V_{z,\theta}^F.
\end{align}
We remark that
\begin{equation}
\underset{\tilde z\rightarrow0}{{\rm Res}}\,\omega_{\frac12,1|0}(\tilde z)=Q_0.
\end{equation}
 In physics, this $Q_0$-dependence can be thought of as a consequence of the so-called ``background charge''.  When $Q_0=0$, \eqref{QFB}-\eqref{QFF} are similar to Eq. (2.36)-(2.38) in \cite{BO2} but they are rather simpler because there is no the involution operator $\sigma$.

Following \cite[Definition 2.8]{BO2}, we define abstract super loop equations:
\begin{definition}\label{def:SLE}
Given an untwisted/$\mu$-twisted super spectral curve $\mathcal{S}_F$, the \emph{untwisted/$\mu$-twisted abstract super loop equations} are the following set of constraints:
\begin{enumerate}
\item \emph{quadratic bosonic loop equations}: for  $2g+n+2m\geq2$ and $(g,n,m)\neq(1,0,0)$
\begin{equation}
	\mathcal{Q}_{g,n+1|2m}^{BB}(z,I|J)+\mathcal{Q}_{g,n+1|2m}^{FF}(z,I|J)\in z^{-N-1} V_z^{B+} \otimes V_z^{B+},
\end{equation}
\item \emph{quadratic fermionic loop equations}: for  $2g+n+2m\geq3$
\begin{equation}
\mathcal{Q}_{g,n|2m}^{BF}(I|z,J)\in z^{-N-1} V_z^{B+}\otimes V_z^{F+}.
\end{equation}
\end{enumerate}
\end{definition}

\subsection{Untwisted/$\mu$-twisted Super Topological Recursion}
We would like to prove that there exists a unique solution of the abstract super loop equations. Similar to the argument in \cite{BO2}, however, existence is easier to prove once we describe them in terms of super Airy structures. Therefore, in this section we assume the existence of a solution, and present a way of computing a unique solution. We call this computational method the ``untwisted/$\mu$-twisted super topological recursion''.

Important objects towards the construction of the super topological recursion are the recursion kernels $K^{BB},K^{BF}$. Given a super spectral curve $\mathcal{S}_F$, we define $K^{BB}$ by
\begin{equation}
K^{BB}(z_0,z)=-\frac{\int^{z}_{0}\omega_{0,2|0}(z_0,\cdot|)}{\omega_{0,1|0}(z|)}.\\
\end{equation}
On the other hand, $K^{BF}$ is defined by
\begin{equation}
K^{BF}(z_0,z)=-\frac{\omega^F_{0,0|2}(|z,z_0)-(f+\frac12)\eta_f(z)\eta_{-f}(z_0)}{2\omega_{0,1|0}(z|)}.
\end{equation}
One should find them similar to the recursion kernels in \cite[Section 3]{BO2}.

With these recursion kernels, we can obtain a unique solution of the abstract super loop equations by the following recursion.

\begin{proposition}\label{prop:STR}
Let  $\tilde{\mathcal{Q}}_{g,n+1|2m}^{BB,FF,BF}$ denote respectively all the terms on the right hand side of \eqref{QFB}-\eqref{QFF} \emph{except} the terms involving $\omega_{0,1|0}$. If there exists a solution to the untwisted/$\mu$-twisted abstract super loop equations that respects the polarization, then it is uniquely constructed recursively by the following formulae:
\begin{itemize}
\item For  $2g+n+2m\geq2$ with $(g,n,m)\neq(1,0,0)$,
\begin{equation}
\omega_{g,n+1|2m}(z_0,I|J)=\underset{z\rightarrow0}{{\rm Res}}\,K^{BB}(z_0,z)\left(\tilde{\mathcal{Q}}_{g,n+1|2m}^{BB}(z,I|J)+\tilde{\mathcal{Q}}_{g,n+1|2m}^{FF}(z,I|J)\right),\label{BR}
\end{equation}
\item For  $2g+n+2m\geq3$,
\begin{equation}
\omega_{g,n|2m+2}(I|u_1,u_2,J)=\hat{\omega}_{g,n|2m+2}(I|u_1,u_2,J)-\eta_{-f}(u_1)\underset{z\rightarrow0}{{\rm Res}}\,\hat{\omega}_{g,n|2m+2}(I|u_2,z,J)\eta_f(z),\label{FR}
\end{equation}
where
\begin{equation}
\hat{\omega}_{g,n|2m+2}(I|u_1,u_2,J)=\underset{z\rightarrow0}{{\rm Res}}\,K^{FB}(u_1,z)\tilde{\mathcal{Q}}_{g,n|2m+2}^{FB}(I|z,u_2,J).\label{FR1}
\end{equation}
\end{itemize}
\end{proposition}

\begin{proof}
The proof closely follows how \cite{BO2} proves the super topological recursion, but our case is simpler because there is no involution operator $\sigma$. Given a super spectral curve $\mathcal{S}_F$, let us assume the existence of a solution of the abstract super loop equations. Since $K^{BB}(z_0,z)$ has a zero of order $N+1$ at $z=0$, the quadratic bosonic loop equations imply that for $2g+n+2m\geq2$ and $(g,n,m)\neq(1,0,0)$,
\begin{equation}
\underset{z\rightarrow0}{{\rm Res}}\,K^{BB}(z_0)\biggl(\mathcal{Q}_{g,n+1|2m}^{BB}(z,I|J)+\mathcal{Q}_{g,n+1|2m}^{FF}(z,I|J)\biggr)=0.
\end{equation}
Let us focus on terms involving $\omega_{0,1|0}$ on the left hand side. They appear in the form:
\begin{align}
\underset{z\rightarrow0}{{\rm Res}}\,K^{BB}(z_0)\biggl(\omega_{0,1|0}(z|)\omega_{g,n+1|2m}(z,I|J)\biggr)&=-\underset{z\rightarrow0}{{\rm Res}}\,\int^{z}_{0}\omega_{0,2|0}(z_0,\cdot|)\omega_{g,n+1|2m}(z,I|J)\nonumber\\
&=-\omega_{g,n+1|2m}(z,I|J)\label{p1},
\end{align}
where we used \eqref{02z=0} in the equality. This proves \eqref{BR}.

Similarly, the quadratic fermionic loop equations imply that  for $2g+n+2m\geq1$, 
\begin{equation}
\underset{z\rightarrow0}{{\rm Res}}\,K^{FB}(u_1,z)\mathcal{Q}_{g,n|2m+2}^{FB}(I|z,u_2,J)=0.\label{p2}
\end{equation}
One can repeat the same procedure as we did in \eqref{p1}. That is, terms involving $\omega_{0,1|0}$ on the left hand side of  \eqref{p2} become
\begin{align}
&-\underset{z\rightarrow0}{{\rm Res}}\left(\omega^F_{0,0|2}(|z,u_1)-\left(f+\frac12\right)\eta_{f}(z)\eta_{-f}(u_1)\right)\omega_{g,n|2m+2}(I|z,u_2,J)\nonumber\\
&=:-\hat{\omega}_{g,n|2m+2}(I|u_1,u_2,,J),
\end{align}
where this should be taken as the definition of $\hat{\omega}_{g,n|2m+2}$. Notice that the only differences between $\hat{\omega}_{g,n|2m+2}(I|u_1,u_2,J)$ and $\omega_{g,n|2m+2}(I|u_1,u_2,J)$ are terms that depend on $\eta_{-f}(u_1)$ thanks to \eqref{NS02expansion} and \eqref{R02expansion}. Since fermionic entries are antisymmetric under their permutations by definition, however, one can indeed supplement this missing $\eta_{-f}(u_1)$-dependence precisely by the second term in \eqref{FR}. It is clear that \eqref{BR} together with \eqref{FR} are recursive for $\omega_{g,n|2m}$ in $2g+n+2m$, hence we have constructed all $\omega_{g,n|2m}$ starting with a super spectral curve, subject to the assumption of the existence of a solution. This completes the proof.

\end{proof}

\begin{remark}
It is highly nontrivial to show that $\omega_{g,n|2m}$  obtained from \eqref{BR} and \eqref{FR} are symmetric under permutations of the first $n$ entries, and anti-symmetric under permutations of the last $2m$ entries.  Also, for $n m\neq0$, there are two formulae to compute $\omega_{g,n|2m}$: from either \eqref{BR} or \eqref{FR}. In other words, having the recursive formulae \eqref{BR} and \eqref{FR} is not sufficient to show existence. We will prove the existence of a solution in Corollary~\ref{coro:main} with the help of super Airy structures.
\end{remark}

\begin{remark}\label{rem:noF0}
It is easy to show that all $\omega_{g,n|2m}(J|K)=0$ for all $2g+n+2m\geq3$ with $g<1$, that is, $g=0$ or $g=\frac12$. This is a general property of the untwisted/$\mu$-twisted super topological recursion. Note that $\omega_{0,n|2m}(J|K)$ can be nonzero if one considers the super topological recursion of \cite{BO2}, i.e., the $\rho$-twisted super topological recursion. 
\end{remark}

\section{Super Airy Structures}\label{sec:SAS}

In this section, we will reformulate the super topological recursion from an algebraic point of view. A key notion is super Airy structures introduced in \cite{SAS}. It has already been shown in \cite{BO2} that the super topological recursion on a $\rho$-\emph{twisted} super spectral curve is dual to a super Airy structure derived from a \emph{$\rho$-twisted} module of the $\mathcal{N}=1$ super Virasoro algebra. Here, our focus is on super Airy structures related to untwisted/$\mu$-twisted modules instead. 

\subsection{Review of Super Airy Structures}

We first review the definition and properties of super Airy structures. This subsection very closely follows \cite[Section 4.1]{BO2}.

Let $U=U_0\oplus U_1\oplus\mathbb{C}^{0|1}$ be a super vector space of dimension $d+1$ over $\mathbb{C}$ (the super vector space could be infinite-dimensional, but for simplicity of presentation we will assume here that it has finite dimension). We define $\{x^i\}_{i\in I}$ to be linear coordinates on $U_0\oplus U_1$ where $I=\{1,...,d\}$ with $x^0$ to be the coordinate of the extra $\mathbb{C}^{0|1}$, and their parity is defined such that $|x^i|=0$ if $x^i\in U_0$, $|x^i|=1$ if $x^i\in U_1$, and $|x^0|=1$. Furthermore, let us denote by
\begin{equation}
\mathcal{D}_{\hbar}(U)=\mathbb{C}\llbracket \hbar^{\frac12},x^0,\hbar\partial_{x^0},\{x^i\}_{i\in I},\{\hbar\partial_{x^i}\}_{i\in I}\rrbracket
\end{equation}
the completed algebra of differential operators acting on $U$ where $\hbar$ is a formal variable. We then introduce a $\mathbb{Z}$-grading by
\begin{equation}
\deg(x^0)=\deg(x^i)=1,\;\;\;\deg(\hbar\partial_{x^0})=\deg(\hbar\partial_{x^i})=1,\;\;\;\;\deg(\hbar)=2.
\end{equation}

\begin{definition}[{\cite[Definition 2.3]{SAS}}]\label{def:SAS}
A \emph{super Airy structure} is a set of differential operators $\{H_i\}_{i\in I}\in\mathcal{D}_{\hbar}(U)$ such that:
\begin{enumerate}
\item for each $i\in I$, $H_i$ is of the form
\begin{equation} 
H_i=\hbar\partial_{x^i}-P_i,\label{form}
\end{equation}
where $P_i\in\mathcal{D}_{\hbar}(U)$ has degree greater than 1 with $|P_i|=|x^i|$,
\item there exists $f_{ij}^k\in\mathcal{D}_{\hbar}(U)$ such that
\begin{equation}
[H_i,H_j]_s=\hbar\sum_{k\in I} f_{ij}^k \,H_k,\label{left}
\end{equation}
where $[\cdot,\cdot]_s$ is a super-commutator.
\end{enumerate}
\end{definition}

It is crucial that the $x^0$-dependence appears only in the $\{P_i\}_{i\in I}$, but not in the degree 1 term (there is no $H_0$). In other words, the dimension of the super vector space $U$ is one more than the number of differential operators $\{H_i\}_{i\in I}$. We thus call $x^0$ the \emph{extra variable}. We note that there is no notion of extra variables in the standard, nonsupersymmetric formalism of Airy structures.

\begin{theorem}[{\cite[Theorem 2.10]{SAS}}]\label{thm:SAS1}
Given a super Airy structure  $\{H_i\}_{i\in I}$, there exists a unique formal power series $\hbar \mathcal{F}(x)\in\mathbb{C}\llbracket \hbar^{\frac12},x^0,(x^i)_{i\in I}\rrbracket$ (up to addition of terms in $\mathbb{C}\llbracket \hbar\rrbracket$) such that:
\begin{enumerate}
\item $\hbar \mathcal{F}(x)$ has no term of degree 2 or less,
\item every term in $\hbar \mathcal{F}(x)$ has even parity,
\item it satisfies $H_i\,e^{\mathcal{F}}=0$.
\end{enumerate}
\end{theorem}

$Z := e^{\mathcal{F}}$ is called the \emph{partition function} and $\mathcal{F}$ \emph{the free energy}. Note that $e^{\mathcal{F}}$ is not a power series in $\hbar$, and so one should replace condition (3) by $e^{-\mathcal{F}}\,H_i\,e^{\mathcal{F}}=0$, which gives a power series in $\hbar$. However, as is standard, we write $H_i\,e^{\mathcal{F}}=0$ for brevity. 

Explicitly, $\mathcal{F}$ can be expanded as follows
\begin{equation}
\mathcal{F}=\sum_{g\in\frac12\mathbb{Z}_{\geq0},\;n\in\mathbb{Z}_{\geq0}}^{2g+n>2}\frac{\hbar^{g-1}}{n!}\sum_{i_1,...,i_n\in \{0,I\}}F_{g,n}(i_1,...,i_n)\prod_{k=1}^nx^{i_k},
\end{equation}
where the restriction on the sum ($2g + n >2$) comes from condition (1) in Theorem~\ref{thm:SAS1}. 
$F_{g,n}(i_1,...,i_n)$ is $\mathbb{Z}_2$-symmetric under permutations of indices\footnote{The original definition in \cite{SAS} considers only power series of $\hbar$ rather than $\hbar^{\frac12}$. However, there is no issue with extending the algebra with $\hbar^{\frac12}$ because $\deg\hbar^{\frac12}=1$. In particular, Theorem~\ref{thm:SAS1} stands because their proof considers induction with respect to $\chi=2g+n$, and $\chi$ remains integers even with $\hbar^{\frac12}$.}.

\subsection{Super Virasoro Type}
We now take both the bosonic and fermionic vector spaces $U_0,U_1$ to be countably infinite dimensional. Also, we explicitly distinguish bosonic and fermionic coordinates, namely, we denote by $\{x^1,x^2,...\}$ and $\{\theta^1,\theta^2,...\}$ the coordinates on $U_0$ and $U_1$ respectively and $\theta^0\in\mathbb{C}^{0|1}$ is treated as the extra variable. In particular, all $\{\theta^0,\theta^1,\theta^2,...\}$ are Grassmann variables. 

We then define differential operators $\{J_a\}_{a\in\mathbb{Z}}$ as
 \begin{equation}
\forall a\in\mathbb{Z}_{>0},\;\;\;\;J_{a}=\hbar\frac{\partial}{\partial x^a},\;\;\;\;J_0=\tau_0+\hbar^{\frac12}Q_0,\;\;\;\;J_{-a}=ax^a.\label{Heis}
\end{equation}
where $\tau_0,Q_0\in\mathbb{C}$ and this is different from $J_0$ in \cite[Section 4.2]{BO2}. For the Neveu-Schwarz sector (equiv. untwisted fermion), we define  Grassmann differential operators $\{\Gamma_r\}_{r\in\mathbb{Z}+\frac12}$ by
\begin{equation}
\forall i\in\mathbb{Z}_{\geq0},\;\;\;\;\Gamma_{i+\frac12}=\hbar\frac{\partial}{\partial\theta^{i}},i\;\;\;\;\Gamma_{-i-\frac12}=\theta^{i},\label{Cliff}
\end{equation}
On the other hand, for the Ramond sector (equiv. $\mu$-twisted fermion), we define $\{\Gamma_r\}_{r\in\mathbb{Z}}$ by
\begin{equation}
\forall r\in\mathbb{Z}_{>0},\;\;\;\;\Gamma_{r}=\hbar\frac{\partial}{\partial\theta^{r}},\;\;\;\;\Gamma_0=\frac{\theta^0}{2}+\hbar\frac{\partial}{\partial\theta^0},\;\;\;\;\Gamma_{-r}=\theta^{r},
\end{equation}
It is straightforward to see that $J_a$ are a basis for the Heisenberg algebra, while the $\Gamma_a$ are a basis for the Clifford algebra:
\begin{equation}
[J_a,J_b]=a\,\hbar\,\delta_{a+b,0},\;\;\;\;[J_a,\Gamma_{b}]=0,\;\;\;\;\{\Gamma_{a},\Gamma_{b}\}=\hbar\,\delta_{a+b,0}.
\end{equation}
Using these differential operators, we define super Virasoro differential operators with background charge $Q_0$ (where $: \cdots :$ denotes normal ordering) as
\begin{align}
n\in\mathbb{Z},\;\;\;\;L_{n}&=\frac12\sum_{j\in\mathbb{Z}}: J_{-j}J_{n+j} : + \frac12\sum_{s\in\mathbb{Z}+f}(\frac n2+s) : \Gamma_{-s}\Gamma_{s+n} :\nonumber\\
&+\delta_{n,0}\delta_{F,R}\frac{\hbar}{16}-\frac{n+1}{2}\hbar^{\frac12}Q_0J_n,\label{L1}\\
r\in\mathbb{Z}+f,\;\;\;\;G_{r}&=\sum_{j\in\mathbb{Z}} :J_{-j}\Gamma_{j+r} :-\left(r+\frac12\right)\hbar^{\frac12}Q_0\Gamma_r,\label{G1}
\end{align}
where recall from \eqref{fandNSR} that $f=1/2$ in the Neveu-Schwarz sector and $f=0$ in the Ramond sector. Note that they generate the $\mathcal{N}=1$ super Virasoro algebra with $\hbar$ inserted \footnote{$\hbar$ is inserted in order to meet the criteria of super Airy structures \eqref{left}. Also, the central charge in \cite{BF,Ito} uses a different normalisation.}:
\begin{align}
[L_m,L_n]=&\hbar(m+n)L_{m+n}+\hbar\frac{c}{12}(m^3-m)\delta_{m+n,0},\\
[L_m,G_r]=&\hbar\left(\frac{m}{2}-r\right)G_{m+r},\\
\{G_r,G_s\}=&2\hbar L_{r+s}+\hbar \frac{c}{3}\left(r^2-\frac14\right)\delta_{r+s,0},\\
c=&\hbar\left(\frac{3}{2}-3Q_0^2\right).
\end{align}
In terms of representations of vertex operator algebras, these differential operators are in untwisted or $\mu$-twisted modules of the $\mathcal{N}=1$ super Virasoro algebra \cite{SAS}.

We now construct a super Airy structure $\tilde{\mathcal{S}}_F$ dual to a super spectral curve $\mathcal{S}_F$. Recall from Remark~\ref{rem:para} that a super spectral curve $\mathcal{S}_F$ is completely determined when one fixes all the parameters $(F,N,\tau_{l},\phi_{kl},\psi_{mn}^F,D_k,Q_l)$. Let us take the same parameters, and consider a set $\tilde{\mathcal{S}}_F=\{H_{i} ,F_{i}\}_{ i\in\mathbb{Z}_{\geq1}}$ of the following differential operators:
\begin{align}
\forall i\in\mathbb{Z}_{\geq1},\;\;\;\;H_{i} =& \Phi_N L_{N+i-1} \Phi_N^{-1}+\hbar \tilde{D}_i,\label{Hi}\nonumber\\
&-\frac12\sum_{k,l=1}^{N-1}\delta_{k+l,N+i-1}(\tau_{-k}+\hbar^{\frac12}Q_{-k})(\tau_{-l}+\hbar^{\frac12}Q_{-l}),\\
F_{i} =& \Phi_N G_{N+i+f-1} \Phi_N^{-1},
\end{align}
where
\begin{equation}
\tilde D_i=\sum_{k=0}^{N-i}\tau_{k-(N-1)}D_{i+k},\;\;\;\;D_{i>N}=0,\label{Di}
\end{equation}
\begin{align}
\Phi_N:=&\exp\left(-\frac{1}{\hbar}\sum_{l=1}^{N-1}\frac{\tau_{-l}+\hbar^{\frac12}Q_{-l}}{l}J_{-l}\right)\nonumber\\
&\times\exp\left(\frac{1}{\hbar}\left(\sum_{l>0}\frac{\tau_l+\hbar^{\frac12}Q_l}{l}J_l+\sum_{l,k>0}\frac{\phi_{kl}}{2kl}J_kJ_l+\sum_{k,l\geq0}\frac{\psi^F_{kl}}{2}\Gamma_{k+f}\Gamma_{l+f}\right)\right).\label{Phi}
\end{align}
Note that the order of operators in $\Phi_N$ is important. That is, the conjugation with negative modes $(J_{-l})_{l\in\mathbb{Z}_{>0}}$ should act after all positive modes $(J_{l})_{l\in\mathbb{Z}_{>0}}$. When $N=1$ and $Q_l=0$ for all $l\in\mathbb{Z}$, $\Phi_1$ is identical to $\Phi$ in \cite[Section 4.2]{BO2}. However, we need to additionally consider conjugation by finitely many negative modes $J_{-l}$ in order to match with the super topological recursion when $N>1$.

\begin{proposition}\label{prop:SAS}
 $\tilde{\mathcal{S}}_F$ forms a super Airy structure.
\end{proposition}
\begin{proof}
Since  $\{L_n,G_r\}$ generates an $\mathcal{N}=1$ super Virasoro subalgebra, and since the $\Phi_N$-action is just conjugate, it is easy to see that $\tilde{\mathcal{S}}_F$ satisfies the second condition in Definition~\ref{def:SAS}. Importantly, since $L_N,...,L_{2N-1+2f}$ never appear on the right hand sides of super Virasoro commutation relations, constant terms can freely be added to, or subtracted from $L_N,...,L_{2N-1+2f}$ without changing commutation relations \footnote{This point is already addressed in \cite[Section 5.1]{SAS} when $\tau_l=\delta_{l,-N+1}$ and $Q_l=\phi_{kl}=\psi_{kl}=0$.}. The second line in \eqref{Hi} is there to remove unwanted constant terms that appear in $ \Phi_N L_{N+i-1} \Phi_N^{-1}$.

What remains to be shown is that there exists a linear transformation that brings  $H_{i} ,F_{i}$ to the form of \eqref{form}, i.e. the first condition in Definition~\ref{def:SAS}. This is proven by explicit evaluations of the conjugation of $L_{N+i-1}$ and $G_{N+i+f-1}$ by $\Phi_N$ which is purely computational. Thus, we give the rest of the computations in Appendix~\ref{app:SAS}.

\end{proof}

It is worth emphasizing that the defining data of a super spectral curve is in a one-to-one correspondence with that of a super Airy structure of Proposition~\ref{prop:SAS}. Thanks to Theorem~\ref{thm:SAS1}, there exists a unique partition function $Z_F$ and free energy $\mathcal{F}_F= \log Z_F$ in the form:
\begin{equation}
\mathcal{F}_F=\sum_{g,n,m\geq0}^{2g+n+2m>2}\frac{\hbar^{g-1}}{n!(2m)!}\sum_{\substack{i_1,...,i_n>1\\j_1,...,j_{2m}\geq0}}F_{g,n|2m}(i_1,...,i_n|j_1,...,j_{2m})\prod_{k=1}^nx^{i_k}\prod_{l=1}^{2m}\theta^{j_l},\label{SASF}
\end{equation}
and such that \footnote{Since $\{\bar H_i,\bar F_i\}$ and $\{ H_i,F_i\}$ are linearly related by some upper triangular matrix, the resulting differential constraints $H_ie^{\mathcal{F}}= F_ie^{\mathcal{F}}=0$ are equivalent to $\bar H_ie^{\mathcal{F}_F}=\bar F_ie^{\mathcal{F}_F}=0$.\label{footnote:equiv}}
\begin{equation}
\forall \;i\in\mathbb{Z}_{>0},\;\;\;\;H_iZ_F=F_iZ_F=0.\label{SVconstraints}
\end{equation}
Note that $F_{g,n|2m}$ is symmetric under permutations of the $n$ first entries, anti-symmetric under permutations of the last $2m$ entries, with no further symmetry. Now, since a choice of parameters $(F,N,\tau_{l},\phi_{kl},\psi_{mn}^F,D_k,Q_l)$ uniquely fixes a pair of a super spectral curve and a super Airy structure, one may ask: are there any relation between $F_{g,n|2m}$ and $\omega_{g,n|2m}$ computed by Proposition~\ref{prop:STR}? The following is the main theorem of the present paper which presents an explicit dictionary between  $F_{g,n|2m}$ and $\omega_{g,n|2m}$:

\begin{theorem}\label{thm:main}
	~
	\begin{enumerate}
		\item
			Consider the super Airy structure $\tilde{\mathcal{S}}_F$ in Proposition \ref{prop:SAS}, defined in terms of parameters $(F,N,\tau_{l},\phi_{kl},\psi_{mn}^F,D_k,Q_l)$. Let
			\begin{equation}
				F_{g,n|2m}(i_1,\ldots , i_n | j_1, \ldots, j_{2m} )
				\end{equation}
				be the coefficients of the associated unique free energy $\mathcal{F}_F$. 
\item
	Let $\mathcal{S}_F$ be the super spectral curve defined with the same parameters. Consider an infinite sequence of multilinear differentials $\omega_{g,n|2m}$ that respect the polarization,
and that satisfy the abstract super loop equations (Definition \ref{def:SLE}). We expand the differentials in terms of the polarised basis as: 
\begin{align}
	\omega_{g,n|2m}(I|J)=\sum_{\substack{i_1,...,i_n>1\\j_1,...,j_{2m}\geq0}}\hat F_{g,n|2m}(i_1,...,i_n|j_1,...,j_{2m})\bigotimes_{k=1}^n d\xi_{-i_k}(z_k)\otimes\bigotimes_{l=1}^{2m}\eta_{-j_l-f}(u_l,\theta_l) .\label{thm}
\end{align}
\end{enumerate}
Then, for all $g,n,m$, and indices $i_1, \ldots, i_n$ and $j_1, \ldots, j_{2m}$,
\begin{equation}
	\hat F_{g,n|2m}(i_1,...,i_n|j_1,...,j_{2m}) = 	F_{g,n|2m}(i_1,\ldots , i_n | j_1, \ldots, j_{2m} ).\label{F=F}
\end{equation}

\end{theorem}

\begin{proof}
Since the proof is based on computations, let us explain here the strategy of the proof and leave the computational details to Appendix~\ref{app:thm}.

Recall that two different sets of equations/constraints have been introduced: one is abstract super loop equations which is defined geometrically (Definition~\ref{def:SLE}), and the other is super Virasoro-like constraints for the partition function of the super Airy structure \eqref{SVconstraints} which is defined algebraically. We then consider:
\begin{itemize}
    \item the set of recursive equations for $\hat{F}_{g,n|2m}$ for each $g,n,m$ which is reduced by manipulations of abstract super loop equations (Definition~\ref{def:SLE}) and by using the relation/definition \eqref{thm} between $\omega_{g,n|2m}$ and $\hat{F}_{g,n|2m}$,
    \item the set of recursive equations for $F_{g,n|2m}$ for each $g,n,m$ which is obtained by explicit evaluation of the super Virasoro-like constraints \eqref{SVconstraints}.
\end{itemize}
After careful yet tedious computations, one finds that the set of recursive equations for $\hat{F}_{g,n|2m}$ is precisely the same as the one for $F_{g,n|2m}$, hence uniqueness of the solution (Theorem~\ref{thm:SAS1}) of super Airy structures implies that $\hat{F}_{g,n|2m}=F_{g,n|2m}$, that is, \eqref{F=F} holds. See Appendix~\ref{app:thm} for the computational details.

\end{proof}

Theorem~\ref{thm:SAS1} and Theorem~\ref{thm:main} immediately show that a unique solution to the untwisted/$\mu$-twisted abstract super loop equations that respects the polarization exists. Thus, we have:

\begin{corollary}\label{coro:main}
	There exists a solution to the untwisted/$\mu$-twisted abstract super loop equations that respects the polarization, and it is uniquely constructed by the untwisted/$\mu$-twisted super topological recursion of Proposition \ref{prop:STR}.
\end{corollary}

\begin{remark}
If one drops all fermionic dependences from discussions, not only $\psi_{kl}$ but also all fermionic modes $\Gamma_i$ and vector spaces $V_z^F$, one gets the topological recursion without branch covers. In particular, Theorem~\ref{thm:main} holds for any positive integer $N$ which is more general than an analogous recursion for $r=2$ in \cite{BBCC}.
\end{remark}

\section{Applications to Superconformal Blocks}\label{sec:SCB}
In the previous section, we presented two dual ways of solving untwisted/$\mu$-twisted abstract super loop equations. In this section, we will apply them to compute the so-called Gaiotto vectors for $\mathcal{N}=1$ superconformal blocks. Thanks to the AGT correspondence, this realises an interesting application of the untwisted/$\mu$-twisted super topological recursion to $\mathcal{N}=2$ pure $U(2)$ supersymmetric gauge theory on $\mathbb{C}^2/\mathbb{Z}_2$.

\subsection{Gaiotto Vectors}
We will define the Gaiotto vector in the Neveu-Schwarz sector and the Ramond sector in this section. Since we would like to adapt the results of \cite{BF,Ito} to connect with four dimensional supersymmetric gauge theories, we follow the presentations of \cite{BF,Ito} in part. In particular, we call the Neveu-Schwarz/Ramond sector instead of the untwisted/$\mu$-twisted sector.

\subsubsection{Gaiotto vector in the Neveu-Schwarz Sector}

Let $\{L_n,G_r\}$ be generators of the $\mathcal{N}=1$ super Virasoro algebra in the Neveu-Schwarz sector of central charge $c$, that is, $n\in\mathbb{Z}$ and $r\in\mathbb{Z}+\frac12$, and we denote by $\mathcal{V}_{\Delta,NS}$ the Verma module of highest weight $\Delta$. Note that since we are interested in the Verma module corresponding to super Liouville field theory \eqref{L1}, \eqref{G1}, $c$ and $\Delta$ are given by
\begin{equation}
c=\hbar\left(\frac32-3Q^2\right),\;\;\;\;\Delta=\frac12(\tau_0^2-\hbar Q^2)\label{weightNS}
\end{equation}
Then, the highest weight state $\ket{\Delta}$ satisfies the following conditions:
\begin{equation}
L_0\ket{\Delta}=\Delta\ket{\Delta},\;\;\;\;\forall n,r>0\;\;\;\;L_n\ket{\Delta}=G_r\ket{\Delta}=0.
\end{equation}
We consider the $L_0$-eigenvalue decomposition $\mathcal{V}_{\Delta,NS}=\bigoplus_{M}\mathcal{V}_{\Delta,NS}^M$ where $M\in\frac12\mathbb{Z}_{\geq0}$, and each $\mathcal{V}_{\Delta,NS}^{M}$ is given as
\begin{equation}
\mathcal{V}_{\Delta,NS}^{M}=\text{Span}\left\{\prod_{i=1}^k\prod_{j=1}^l L_{-n_i}G_{-r_j}\ket{\Delta}\right\},
\end{equation}
where
\begin{equation}
n_1\geq\cdots\geq n_k>0,\;\;\;\;r_1>\cdots> r_l>0,\;\;\;\;\sum_{i=1}^kn_i+\sum_{j=1}^l r_j=M
\end{equation}

For $M\in\frac12\mathbb{Z}_{\geq0}$, let us assume that there exists a set of vectors $\ket{M}\in\mathcal{V}_{\Delta,NS}^{M}$ satisfying
\begin{equation}
L_1\ket{M}=\ket{M-1},\;\;\;\;\forall n, r>1\;\;\;\;L_n\ket{M}=G_r\ket{M}=0,\label{WhittakerNS}
\end{equation}
where $\ket{0}=\ket{\Delta}$. Note that we do \emph{not} impose anything for the action of $G_{\frac12}$. However, if one defines  another set of vectors $\widetilde{\ket{M}}\in\mathcal{V}_{\Delta,NS}^M$ by
\begin{equation}
\widetilde{\ket{M-\frac12}}:=G_{\frac12}\ket{M},
\end{equation}
then it is easy to show from \eqref{WhittakerNS} that
\begin{equation}
L_1\widetilde{\ket{M}}=\widetilde{\ket{M-1}},\;\;\;\;\forall n, r>1\;\;\;\;L_n\widetilde{\ket{M}}=G_r\widetilde{\ket{M}}=0,\label{WhittakerNS1}
\end{equation}
Therefore, \eqref{WhittakerNS} is equivalent to \cite[Eq. (3.15)]{BF}

 We are now ready to define the Gaiotto vector in the Neveu-Schwarz sector:
\begin{definition}[{\cite[Section 3.1]{BF}}]
Let us assume that there exists a set of vectors $\{\ket{M}\}\in\mathcal{V}_{\Delta,NS}^{M}$ satisfying \eqref{WhittakerNS}. Then, for a formal variable $\Lambda\in\mathbb{C}$, the ``Gaiotto vector in the Neveu-Schwarz sector'' $\ket{G}_{NS}\in\mathcal{V}_{\Delta,NS}$ is defined by
\begin{equation}
\ket{G}_{NS}:=\sum_{M\in\frac12\mathbb{Z}_{\geq0}}\Lambda^{2M}\ket{M}.
\end{equation}
\end{definition}

The Gaiotto vector naturally arises in the so-called ``Gaiotto limit'', ``Whittaker limit'', or ``irregular limit'' in the context of superconformal blocks. We refer to the readers \cite{BF,Ito} and references therein for further details.  Notice that one can show from \eqref{WhittakerNS} that the Gaiotto vector satisfies:
\begin{equation}
L_1\ket{G}_{NS}=\Lambda^2\ket{G}_{NS},\;\;\;\;\forall n,r>1\;\;\;\;L_n\ket{G}_{NS}=G_r\ket{G}_{NS}=0.\label{WhittakerNS2}
\end{equation}
Indeed, one can alternatively take \eqref{WhittakerNS2} as a definition of the Gaiotto vector\footnote{Without supersymmetry, a vector $\ket{w}$ in the Verma module satisfying $L_1 \ket{w}=\Lambda\ket{w}$ and $L_{n\geq2}\ket{w}=0$ is called the ``Whittaker vector'' whereas the corresponding Gaiotto vector is rather defined as a formal sum of some cohomology classes of an appropriate instanton moduli space \cite[Section 2.1]{BBCC}. Thus, strictly speaking, one may call $\ket{G}$ in Definition~\ref{def:GaiottoNS} the Whittaker vector instead of the Gaiotto vector, though the AGT correspondence states that these two vectors are equivalent to each other. In the present paper, however, we call it the Gaiotto vector in order to emphasise a relation to four-dimensional supersymemtric gauge theory.}.
\begin{definition}\label{def:GaiottoNS}
For $\{\ket{M}\}\in\mathcal{V}_{\Delta,NS}^{M}$ for every $M\in\frac12\mathbb{Z}$, we consider a vector $\ket{G}_{NS}\in\mathcal{V}_{\Delta,NS}$ as a formal power series of $\Lambda$ by
\begin{equation}
\ket{G}_{NS}:=\sum_{M\in\frac12\mathbb{Z}_{\geq0}}\Lambda^{2M}\ket{M}.
\end{equation}
Then, $\ket{G}_{NS}$ is said to be the ``Gaiotto vector in the Neveu-Schwarz sector'' if it satisfies
\begin{equation}
L_1\ket{G}_{NS}=\Lambda^2\ket{G}_{NS},\;\;\;\;\forall n,r>1\;\;\;\;L_n\ket{G}_{NS}=G_r\ket{G}_{NS}=0.\label{WhittakerNS3}
\end{equation}
\end{definition}
One can easily derive \eqref{WhittakerNS} from \eqref{WhittakerNS3} order by order in $\Lambda$. We call a vector $\ket{M}$ satisfying \eqref{WhittakerNS} the ``Gaiotto vector of level $M$ in the Neveu-Schwarz sector''. We note that existence of such a vector is not \emph{a priori} guaranteed.

\subsubsection{Gaiotto Vector in the Ramond Sector}

Let $\{L_n,G_r\}$ be generators of the $\mathcal{N}=1$ super Virasoro algebra in the Ramond sector of central charge $c$, that is, $n,r\in\mathbb{Z}$, and we denote by $\mathcal{V}_{\Delta,R}$ the Verma module of highest weight $\Delta$. Note that in the Ramond sector, $c$ and $\Delta$ are given by
\begin{equation}
c=\hbar\left(\frac32-3Q^2\right),\;\;\;\;\Delta=\frac{\hbar}{16}+\frac12(\tau_0^2-\hbar Q^2)\label{weightR}
\end{equation}
Unlike the Neveu-Schwarz sector, there are two highest weight states $\ket{\Delta}_{\pm}$ satisfying the following conditions:
\begin{equation}
L_0\ket{\Delta}_{\pm}=\Delta\ket{\Delta}_{\pm},\;\;\;\;G_0\ket{\Delta}_{\pm}=\sqrt{\Delta-\frac{c}{24}}\ket{\Delta}_{\mp},\;\;\;\;\forall n,r>0\;\;\;\;L_n\ket{\Delta}_{\pm}=G_r\ket{\Delta}_{\pm}=0,
\end{equation}
We consider the $L_0$-eigenvalue decomposition $\mathcal{V}_{\Delta,R}=\bigoplus_{M}\mathcal{V}_{\Delta,R}^M$ where $M\in\mathbb{Z}_{\geq0}$ in the Ramond sector, and each $\mathcal{V}_{\Delta,R}^{M}$ is given as
\begin{equation}
\mathcal{V}_{\Delta,R}^{M}=\text{Span}\left\{\prod_{i=1}^k\prod_{j=1}^l L_{-n_i}G_{-r_j}\ket{\Delta}_{\pm}\right\},
\end{equation}
where
\begin{equation}
n_1\geq\cdots\geq n_k>0,\;\;\;\;r_1>\cdots> r_l>0,\;\;\;\;\sum_{i=1}^kn_i+\sum_{j=1}^l r_j=M
\end{equation}

We now define the Gaiotto vectors in the Ramond sector similar to Definition~\ref{def:GaiottoNS}:
\begin{definition}[{\cite[Section 3]{Ito}}]\label{def:GaiottoR}
For $\{\ket{M}_{\pm}\}\in\mathcal{V}_{\Delta,R}^{M}$ for every $M\in\mathbb{Z}$, we consider two vectors $\ket{G}_{R\pm}\in\mathcal{V}_{\Delta,R}$ as formal power series of $\Lambda$ by
\begin{equation}
\ket{G}_{R\pm}:=\sum_{M\in\mathbb{Z}_{\geq0}}\Lambda^{2M}\ket{M}_{\pm}.
\end{equation}
Then, $\ket{G}_{R\pm}$ are said to be the ``Gaiotto vectors in the Ramond sector'' if it satisfies
\begin{equation}
L_1\ket{G}_{R\pm}=\frac{\Lambda^2}{2}\ket{G}_{R\pm},\;\;\;\;G_1\ket{G}_{R\pm}=0,\;\;\;\;\forall n,r>1\;\;\;\;L_n\ket{G}_{R\pm}=G_r\ket{G}_{R\pm}=0.\label{WhittakerR3}
\end{equation}
\end{definition}
Note that there are two Gaiotto vectors $\ket{G}_{R\pm}$ and they encode exactly the same information. We call $\ket{G}_{R+}$ the ``bosonic Gaiotto vector'' and $\ket{G}_{R-}$ the ``fermionic Gaiotto vector'', respectively. It is straightforward to show that $\ket{M}_{\pm}$ in the Gaiotto vectors satisfy:
\begin{equation}
L_1\ket{M}_{\pm}=\frac12\ket{M-1}_{\pm},\;\;\;\;G_1\ket{M}_{\pm}=0,\;\;\;\;\forall n,r>1\;\;\;\;L_n\ket{M}_{\pm}=G_r\ket{M}_{\pm}=0.\label{WhittakerR}
\end{equation}
We call a vector $\ket{M}_{+(-)}$ satisfying \eqref{WhittakerR} the ``bosonic (fermionic) Gaiotto vector of level $M$ in the Ramond sector'' \footnote{We abuse the notation and $\ket{M}$ without any subscript refers to the Gaiotto vector of level $M$ in the Neveu-Schwarz sector while $\ket{M}_{\pm}$ with subscript are in the Ramond sector.}. We again note that existence of such vectors is not \emph{a priori} guaranteed.

\subsection{Nekrasov Partition Function}
Let us now briefly review a conjectural relation to $\mathcal{N}=2$ pure $U(2)$ gauge theory on $\mathbb{C}^2/\mathbb{Z}_2$. See \cite{BF,BMT,Ito,Itothesis} for further details.  For more general perspectives of the AGT correspondence, we refer to the readers \cite{LF,T1,T2}.

The Nekrasov partition function of pure $U(2)$ theory on $\mathbb{C}^2/\mathbb{Z}_2$ depends on three parameters $(\epsilon_1,\epsilon_2,a)$, similar to pure $U(2)$ theory on $\mathbb{C}^2$ where $\epsilon_1,\epsilon_2$ are the equivariant parameters of the $(\mathbb{C}^*)^2$-action and $\pm a$ are the eigenvalues of the vector multiplet scalar in the Coulomb branch. However, there are two distinct features about gauge theory on $\mathbb{C}^2/\mathbb{Z}_2$.

One of them is nontrivial holonomies. Let $A$ be a flat $U(2)$-connection of the gauge theory. Since the asymptotic region of $\mathbb{C}^2/\mathbb{Z}_2$ is isomorphic to $S^3/\mathbb{Z}_2$, there are noncontractable cycles, and the holonomy of a cycle in that region
\begin{equation}
U=\exp\left(i \oint A\right)
\end{equation}
satisfies $U^2=1$. Thus, there are four inequivalent classes of holonomies as
\begin{equation}
U=\{\text{diag}(1,1),\text{diag} (1,-1), \text{diag}(-1,1), \text{diag}(-1,-1)\}.
\end{equation}
As a consequence, the gauge theory admits \emph{two} sectors which we call the Neveu-Schwarz and Ramond sector due to the correspondence to superconformal blocks stated shortly. The holonomies of type $(1,1)$ and $(-1,-1)$ contribute to the Neveu-Schwarz sector whereas the Ramond sector is described by the holonomies of type $(1,-1)$ and $(-1,1)$. Note that the Ramond sector (holonomy of type $(1,-1)$ and $(-1,1)$) does not exist in $SU(2)$ theory. 

Another difference is about instanton moduli spaces. Since the $\mathbb{Z}_2$-action sends $(z_1,z_2)\in\mathbb{C}^2\mapsto(-z_1,-z_2)$, the path integral is computed by summing over the space of the $\mathbb{Z}_2$-invariant field configurations on $\mathbb{C}^2$. That is, the instanton moduli spaces for $U(2)$ gauge theory on $\mathbb{C}^2/\mathbb{Z}_2$ are appropriate $\mathbb{Z}_2$-symmetric subspaces of those for $U(2)$ gauge theory on $\mathbb{C}^2$. See \cite{BF,Ito} for practical computations of the Nekrasov partition function in terms of Young tableaux.

Let us now state the conjecture given in \cite{BF,Ito}\footnote{\cite{Ito} takes the BPZ conjugation but we use the standard Hermite conjugation. In particular, there is no $(-i)$ in \eqref{AGTR} in our notation.}:
\begin{conjecture}[\cite{BF,Ito}]\label{conj:AGT}
Let $Z_{{\rm Nek}}^{2M,(q_1,q_2)}$ be the $2M$-instanton contributions to the Nekrasov partition function for $\mathcal{N}=2$ pure $U(2)$ gauge theory on $\mathbb{C}^2/\mathbb{Z}_2$ of holonomy type $(q_1,q_2)$ where $q_1,q_2\in\{\pm1\}$. Also, let $\ket{M}, \ket{M}_{\pm}$ be the Gaiotto vectors of level $M$ in the Neveu-Schwarz sector and Ramond sector respectively. Then, they satisfy:
\begin{align}
\forall M\in\mathbb{Z}_{\geq0},\;\;\;\;\;\;\;\;Z_{{\rm Nek}}^{2M,(1,1)}&=\braket{M|M},\\
\forall M\in\mathbb{Z}_{\geq0}+\frac12,\;\;\;\;Z_{{\rm Nek}}^{2M,(-1,-1)}&=\braket{M|M},\\
\forall M\in\mathbb{Z}_{\geq0},\;\;\;\;\;\;Z_{{\rm Nek}}^{2M,(1,-1)}&=Z_{{\rm Nek}}^{2M,(-1,1)}=\,_+\braket{M|M}_+=\,_-\braket{M|M}_-.
\end{align}
Equivalently, let $Z_{{\rm Nek}}^{F}$ be the Nekrasov partition function for $\mathcal{N}=2$ pure $U(2)$ gauge theory on $\mathbb{C}^2/\mathbb{Z}_2$ with $F\in\{NS,R\}$, that is,
\begin{align}
&Z_{{\rm Nek}}^{NS}=\sum_{M\in\mathbb{Z}_{\geq0}}\Lambda^{4M}Z_{{\rm Nek}}^{2M,(1,1)}+\sum_{M\in\mathbb{Z}_{\geq0}+\frac12}\Lambda^{4M}Z_{{\rm Nek}}^{2M,(-1,-1)},\\
&Z_{{\rm Nek}}^{R}=\sum_{M\in\mathbb{Z}_{\geq0}}\Lambda^{4M}Z_{{\rm Nek}}^{2M,(1,-1)}=\sum_{M\in\mathbb{Z}_{\geq0}}\Lambda^{4M}Z_{{\rm Nek}}^{2M,(-1,1)},
\end{align}
Also, let $\ket{G}_{NS},\ket{G}_{R\pm}$ be the Gaiotto vectors in the Neveu-Schwarz sector and Ramond sector respectively. Then, they satisfy:
\begin{align}
Z_{{\rm Nek}}^{NS}&=\,_{NS}\braket{G|G}_{NS},\label{AGTNS}\\
Z_{{\rm Nek}}^{R}&=\,_{R+}\braket{G|G}_{R+} =\,_{R-}\braket{G|G}_{R-},\label{AGTR}
\end{align}
where the parameters $(\epsilon_1,\epsilon_2,a)$ and $(\hbar,Q_0,\tau_0)$ are identified as
\begin{equation}
-\epsilon_1\epsilon_2=\hbar,\;\;\;\;\epsilon_1+\epsilon_2=\hbar^{\frac12}\,Q_0,\;\;\;\;a=\tau_0.
\end{equation}
\end{conjecture}

\subsection{Gaiotto Vectors from Super Airy Structures}

Conjecture~\ref{conj:AGT} does not impose on how one should represent super Virasoro operators $\{L_n,G_r\}$. In this section, we represent $\{L_n,G_r\}$ as differential operators as in \eqref{L1} and \eqref{G1}, and we prove that the partition function of an appropriate super Airy structure becomes the Gaiotto vector in this representation, after a change of parameters. Before doing so, however, let us show some properties of the partition function of the super Airy structure with the following choice of parameters:
\begin{equation}
\biggl(F,N,\tau_{l},\phi_{kl},\psi_{mn}^F,D_k,Q_l\biggr)=\biggl(F,1,\alpha\delta_{l,0},0,0,\frac{1+2f}{2}T\delta_{k,1},Q\delta_{l,0}\biggr)\label{STRG}
\end{equation}

When we take parameters as above, some properties of the free energy $\mathcal{F}_F$ can be checked explicitly. In particular, we can analyse the $T$-dependence of $\mathcal{F}_F$\footnote{This is analogous to \cite[Lemma 4.5]{BBCC}}: 

\begin{lemma}\label{lem:power}
Let $F_{g,n|2m}(I|J)$ be the coefficients of the free energy $\mathcal{F}_F$ associated with the super Airy structure with the above choice of parameters. Then, they satisfy:
\begin{description}
\item[1]  the $T$-dependence of $F_{g,n|2m}(I|J)$ is factored as
\begin{equation}
F_{g,n|2m}(I|J)=T^{i_1+\cdots i_n+j_1+\cdots j_{2m}+2fm} \tilde{F}_{g,n|2m}(I|J)
\end{equation}
where $\tilde{F}_{g,n|2m}(I|J)$ is independent of $T$.
\item[2] $F_{g,n|2m}(I|J)=0$ for $i_1+\cdots i_n+j_1+\cdots j_{2m}+2fm>g$.
\end{description}
\end{lemma}
\begin{proof}
We prove by induction in $2g+n+2m\geq3$. Since all we have to do is to look at \eqref{BSAS} and \eqref{FSAS} with the choice of parameters \eqref{STRG}, and since the proof is based on simple computations, we only give a sketch. When $2g+n+2m=3$, one finds from \eqref{BSAS} and \eqref{FSAS} (or from Proposition~\ref{prop:STR}) that
\begin{align}
&\forall i_1\in\mathbb{Z}_{>0}\;\;\;\;\;\;\;\,\,\,\,\,\,\,\,\,\,\,\,F_{1,1|0}(i_1|)=\frac{1+2f}{2}T\delta_{i_1,1},\\
&\forall i_1,i_2\in\mathbb{Z}_{>0}\;\;\;\;\,\,F_{\frac12,2|0}(i_1,i_2|)=0,\\
&\forall j_1,j_2\in\mathbb{Z}_{\geq0}\;\;\;\;F_{\frac12,0|2}(|j_1,j_2)=0.
\end{align}
Thus, it holds when $2g+n+2m=3$. Note that one can easily show $F_{g,n|2m}(I|J)=0$ whenever $g<1$ as discussed in Remark~\ref{rem:noF0}.

Let us now assume the above statements hold for all $g',n',m'$ whenever $3\leq 2g'+n'+2m'\leq2g+n+2m$. Then, $F_{g,n+1|2m}(i_0,I|J)$ can be computed by \eqref{BSAS}, and one notices that due to the Kronecker delta's in \eqref{C1} and \eqref{C2}, each term gives the same power of $T$, namely, $T^{i_0+i_1+\cdots i_n+j_1+\cdots j_{2m}+2fm}$. The Kronecker delta's also guarantee by induction that $F_{g,n+1|2m}(I|J)=0$ for all $i_0+i_1+\cdots i_n+j_1+\cdots j_{2m}+2fm>g$. Note that $F_{g,0|2\tilde{m}}(|j_1,\tilde J)$ with $2g+n+2m<2\tilde m$ cannot be computed by \eqref{BSAS}, hence, we have to check it separately. Nevertheless, similar to the previous case, the Kronecker delta in \eqref{C3} also implies that  each term in \eqref{FSAS} is a monomial in $T$ of degree $T^{j_1+\cdots j_{2\tilde m}+2f \tilde m}$. This completes the proof.
\end{proof}

We now consider a change of parameters from $T$ in \eqref{STRG} to $\Lambda$ in Conjecture~\ref{conj:AGT}. Namely, we consider
\begin{equation}
\Lambda^2=\hbar T.\label{LamtoT}
\end{equation}
Since this modifies the powers of $\hbar$ in each term, let us see how $\mathcal{F}_F(\Lambda^2,\hbar)$ behaves. In particular, we would like to find the leading order in $\hbar$ after the change of parameters.

Thanks to Lemma~\ref{lem:power}, we are able to rewrite the free energy as follows:
\begin{align}
\mathcal{F}_{F}=&\sum_{g,n,m\geq0}^{2g+n+2m>2}\frac{\hbar^{g-1}}{n!(2m)!}\sum_{\substack{i_1,...,i_n>1\\j_1,...,j_{2m}\geq0}}F_{g,n|2m}(I|J)\prod_{k=1}^nx^{i_k}\prod_{l=1}^{2m}\theta^{j_l}\nonumber\\
=&\sum_{h,n,m\geq0}\frac{\hbar^{h-1}}{n!(2m)!}\sum_{\substack{i_1,...,i_n>1\\j_1,...,j_{2m}\geq0}}\Lambda^{2(i_1+\cdots i_n+j_1+\cdots j_{2m}+2fm)}\Phi_{h,n|2m}(I|J)\prod_{k=1}^nx^{i_k}\prod_{l=1}^{2m}\theta^{j_l},\label{powerLambda}
\end{align}
where
\begin{equation}
\Phi_{h,n|2m}(I|J)=F_{h+(i_1+\cdots i_n+j_1+\cdots j_{2m}+2fm),n|2m}(I|J)
\end{equation}
Two important remarks are in order. First, the leading order in $\hbar$ is still of order $\hbar^{-1}$ and $\hbar\,\mathcal{F}_F(\Lambda^2,\hbar)$ after the change of parameters is still a power series in $\hbar^{\frac12}$. Second, $\Phi_{h,n|2m}\neq0$ even for $2g+n+2m<3$ unlike $F_{g,n|2m}$ due to the change of parameters \eqref{LamtoT}.

With this under our belt, we show that the Gaiotto vector in the Neveu-Schwarz or the Ramond sector corresponds to the partition function of the super Airy structure with parameters determined by \eqref{STRG}:

\begin{proposition}\label{prop:main}
Consider the super Airy structure $\tilde{\mathcal{S}}_F$ in Proposition \ref{prop:SAS} with the parameters set by \eqref{STRG}, and let $Z_F$ be the unique partition function of $\tilde{\mathcal{S}}_F$. Then, for $F=NS$, $Z_{NS}$  becomes the Gaiotto vector $\ket{G}$ in the Neveu-Schwarz sector, and for $F=R$, $Z_R$ becomes the bosonic Gaiotto vector $\ket{G}_{R+}$ in the Ramond sector after the change of parameters $\Lambda=\hbar T$.
\end{proposition}
\begin{proof}
We first consider the Neveu-Schwarz sector. By construction, $Z_{NS}$ satisfies:
\begin{equation}
H_1Z_{NS}=(L_1-\Lambda^2)Z_{NS}=0,\;\;\;\;\forall i\geq1\;\;\;\;H_{i+1}Z_{NS}=L_{i+1}Z_{NS}=0,\;\;\;\;F_iZ_{NS}=G_{i+\frac12}Z_{NS}=0.
\end{equation}
where we used $\Lambda^2=\hbar T$.

Second, \eqref{powerLambda} implies that the free energy $\mathcal{F}_{NS}$ is a power series in $\Lambda$. More precisely,
\begin{equation}
\mathcal{F}_{NS}\in\Lambda^2\,\mathbb{C}\llbracket \Lambda^2\rrbracket.
\end{equation}
This shows that the leading term of the partition function is exactly 1:
\begin{equation}
Z_{NS}=e^{\mathcal{F}_{NS}}=1+\mathcal{O}(\Lambda^2).
\end{equation}

At last, we will need to identify $1$ in this representation with the highest weight vector $\ket{0}=\ket{\Delta}$ in the Neveu-Schwarz sector. This is indeed straightforward to check from \eqref{Heis}, \eqref{Cliff}, \eqref{L1}, and \eqref{G1}. Indeed for the Neveu-Schwarz sector, we have:
\begin{equation}
L_0\cdot1=\Delta\cdot1,\;\;\;\;\forall n,r>0\;\;\;\;L_n\cdot1=G_r\cdot1=0,
\end{equation}
where $\Delta$ coincides with \eqref{weightNS}. Therefore, the proposition holds for the Neveu-Schwarz sector.

Next, we turn to the Ramond sector. Notice that since the partition function of any super Airy structure is necessarily bosonic,  we only consider $\ket{G}_{R+}$. Then, the rest goes parallel to the argument for the Neveu-Schwarz sector. Namely,
\begin{equation}
H_1Z_{R}=\left(L_1-\frac12\Lambda^2\right)Z_{R}=0,\;\;\;\;\forall i\geq1\;\;\;\;H_{i+1}Z_{R}=L_{i+1}Z_{R}=0,\;\;\;\;F_iZ_{R}=G_{i}Z_{R}=0,
\end{equation}
with the identification $\Lambda=\hbar T$. And it can be shown that
\begin{equation}
Z_{R}=e^{\mathcal{F}_R}=1+\mathcal{O}(\Lambda^2),
\end{equation}
hence, we expect $\ket{\Delta}_+=1$. At last, since the representation of $L_n,G_r$ in the Ramond sector is different from the one in the Neveu-Schwarz sector, we find 
\begin{equation}
L_0\cdot1=\Delta\cdot1,\;\;\;\;\forall n,r>0\;\;\;\;L_n\cdot1=G_r\cdot1=0,
\end{equation}
with $\Delta$ given in \eqref{weightR}. This completes the proof.
\end{proof}

Since Theorem~\ref{thm:SAS1} guarantees existence of a unique solution of a super Airy structure, we obtain the following corollary:
\begin{corollary}
The Gaiotto vector in the Neveu-Schwarz sector and the bosonic Gaiotto vector in the Ramond sector exist. And they can be computed by the untwisted or $\mu$-twisted super topological recursion.
\end{corollary}

\begin{remark}
The relation between the bosonic and fermionic Gaiotto vectors $\ket{G}_{R\pm}$ in the Ramond sector is highly nontirivial, and the general relation is not known\footnote{\cite{Ito} computed up to level 2.}. As a consequence, existence of the fermionic Gaiotto vector is not supported by the discussions of the present paper. Note, however, that at the zero level, the fermionic highest weight state $\ket{0}_-=\ket{\Delta}_-$ in our representation is given by the extra variable $\theta^0$:
\begin{equation}
\ket{\Delta}_-=\frac{G_0\cdot1}{\sqrt{\Delta-c/24}}=\sqrt{2}\theta^0.
\end{equation}
\end{remark}

\subsubsection{Conjugate Operators}

Let us consider the Hermitian conjugate for the Verma module. This means that
\begin{equation}
\forall i\in\mathbb{Z}_{\neq0}\;\;\;\;(J_i)^{\dagger}=J_{-i},\;\;\;\;\forall r\in\mathbb{Z}+f\;\;\;\;(\Gamma_{r})^{\dagger}=\Gamma_{-r}.
\end{equation}
In order to split the zero mode $\Gamma_0$ into $\theta^0$ and $\hbar\partial_{\theta^0}$, we use the following notation:
\begin{equation}
\tilde \Gamma_0:=\frac12 \theta^0,\;\;\;\;(\tilde \Gamma_0)^{\dagger}:=\hbar\partial_{\theta^{0}},\;\;\;\;\forall i\in\mathbb{Z}_{\neq0}\;\;\;\;\tilde \Gamma_{i+f}=\Gamma_{i+f}
\end{equation}
Let us then rewrite \eqref{powerLambda} in terms of these modes instead of variables as follows:
\begin{align}
\mathcal{F}_F=&\sum_{g,n,m\geq0}^{2g+n+2m\geq3}\frac{\hbar^{g-1}}{n!(2m)!}\sum_{\substack{i_1,...,i_n>1\\j_1,...,j_{2m}\geq0}}F_{g,n|2m}(I|J)\prod_{k=1}^n\frac{J_{-i_k}}{i_k}\prod_{l=1}^{2m}(1+\delta_{j_l+f,0})\tilde \Gamma_{-j_l}\nonumber\\
=&\sum_{h,n,m\geq0}\frac{\hbar^{h-1}}{n!(2m)!}\sum_{\substack{i_1,...,i_n>1\\j_1,...,j_{2m}\geq0}}\Lambda^{2(i_1+\cdots i_n+j_1+\cdots j_{2m}+2fm)}\Phi_{h,n|2m}(I|J)\prod_{k=1}^n\frac{J_{-i_k}}{i_k}\prod_{l=1}^{2m}(1+\delta_{j_l+f,0})\tilde \Gamma_{-j_l},\label{powerLambda1}
\end{align}
where the $(1+\delta_{j_l+f,0})$ is inserted to cancel out the $\frac12$ in the definition of $\tilde\Gamma_0$. Hence, by construction, we find that $\bra{G}$ becomes a differential operator in this representation:
\begin{equation}
(Z_F)^{\dagger}=e^{(\mathcal{F}_F)^{\dagger}}.
\end{equation}
As a consequence, the norm $(\cdot|\cdot)$ is defined by
\begin{equation}
(Z_F|Z_F)=(Z_F)^{\dagger}\cdot Z_F\bigr|_{x=\theta=0}.\label{norm}
\end{equation}

Therefore, Proposition~\ref{prop:main} shows that \eqref{AGTNS} and \eqref{AGTR} in Conjecture~\ref{conj:AGT} are extended by super Airy structures as follows:
\begin{align}
Z_{{\rm Nek}}^{NS}&=\,_{NS}\braket{G|G}_{NS}=(Z_{NS}|Z_{NS}),\\
Z_{{\rm Nek}}^{R}&=\,_{R+}\braket{G|G}_{R+} =\,_{R-}\braket{G|G}_{R-}=(Z_R|Z_R)
\end{align}

\subsubsection{Graphical Interpretation}
\eqref{powerLambda1} suggests that one can compute the free energy $\mathcal{F}_F$ by summing over connected graphs of appropriate weights. For $h,n,m\in\mathbb{Z}_{\geq0}$ with $2h+n+2m\geq1$, let $\gamma_{h,n|2m}(I|J)$ be the connected planar graph of:
\begin{enumerate}
\item an $(n+2m)$-valent vertex which carries a nonnegative integer $h$,
\item  $n$ bosonic edges whose external vertices are labelled clockwise by $I=(i_1,...,i_n)$,
\item $2m$ fermionic edges whose external vertices are labelled clockwise by $J=(j_1,...,j_{2m})$.
\end{enumerate}
See Figure~\ref{fig:graph1} below. We call $h$ the ``number of loops'' for convention, though it is just an integer associated with an $(n+2m)$-valent vertex. We also denote by $\mathbb{G}^{\text{conn}}$ the set of all such connected graphs. That is, every graph in $\mathbb{G}^{\text{conn}}$ is uniquely determined by a nonnegative integer $h$, a set of positive integers $I$, and a set of nonnegative integers $J$ with $2h+n+2m\geq1$.
\begin{figure}[h]
\begin{equation}
\gamma_{h,n|2m}(I|J)=
\vcenter{\hbox{
\begin{tikzpicture}
\draw[color=red, dashed](0,0) -- (-1,0.5);
\draw[color=red, dashed](0,0) -- (-1,-0.5) ;
\node(i1) at (-1.3,0.5) {$i_n$};
\node(dot) at (-1,0.05) {$\vdots$};
\node(in) at (-1.3,-0.5) {$i_1$};
\draw[color=red] (0,0) -- (1,0.5);
\draw[color=red] (0,0) -- (1,-0.5) ;
\node(j1) at (1.3,0.5) {$j_{1}$};
\node(dots) at (1,0.05) {$\vdots$};
\node(j2m) at (1.45,-0.5) {$j_{2m}$};
\node(n1) at (1,0.5) {$\bullet$};
\node(n2) at (1,-0.5) {$\bullet$};
\filldraw[color=red!60, fill=red!5, very thick](0,0) circle (0.35);
\node(0,0) {{\large $h$}};
\filldraw[color=black!100, fill=black!0](-1,0.5) circle (0.08);
\filldraw[color=black!100, fill=black!0](-1,-0.5) circle (0.08);
\end{tikzpicture}}}
\end{equation}
\caption{Pictorial representation of $\gamma_{h,n|2m}(I|J)$. $\circ$ ($\bullet$) denotes bosonic (fermionic) vertices, and dashed (solid) lines are bosonic (fermionic) edges.}\label{fig:graph1}
\end{figure}

From \eqref{powerLambda1}, the weigh $w$ of $\gamma_{h,n|2m}(I|J)$ is defined as
\begin{equation}
w\left(\gamma_{h,n|2m}(I|J)\right)=\hbar^{h-1}\Lambda^{2(i_1+\cdots i_n+j_1+\cdots j_{2m}+2fm)}\Phi_{h,n|2m}(I|J)\prod_{k=1}^n\frac{J_{-i_k}}{i_k}\prod_{l=1}^{2m}(1+\delta_{j_l+f,0})\tilde \Gamma_{-j_l}.\label{weightsred}
\end{equation}
As a note, $h$ appears on the exponent of $\hbar$ and this is why $h$ is called the number of loops. In addition, the symmetry factor $|\cdot|$ of $\gamma_{h,n|2m}(I|J)$ is given, by definition, as
\begin{equation}
|\gamma_{h,n|2m}(I|J)|=n!(2m)!.
\end{equation}
Then, \eqref{powerLambda1} can be written as
\begin{equation}
\frac{\mathcal{F}_F}{\hbar}=\sum_{\gamma\in\mathbb{G}^{\text{conn}}}\frac{w(\gamma)}{|\gamma|}\label{Fdiag}
\end{equation}
Moreover, since the partition function $Z_F$ is the exponential of $\mathcal{F}_F/\hbar$, it is computed by summing over both connected and disconnected graphs which we simply denote by $\mathbb{G}$:
\begin{equation}
Z_F=\sum_{\gamma\in\mathbb{G}}\frac{w(\gamma)}{|\gamma|},\label{Zgraph}
\end{equation}
where the weight of a disconnected graph is defined to be the product of weights of connected components. The symmetry factor of a disconnected graph is defined in a canonical way, that is, it is the product of symmetry factors of connected components times the product of factorials of the multiplicity of each component.

One can repeat the same steps for $(Z_F)^{\dagger}$. Let us denote by $\mathbb{G}'$ be the set of all graphs in $\mathbb{G}$ but with a different colour and with the opposite order of labelling (see Figure~\ref{fig:graph2} below). We similarly define the weight and the symmetry factor of a connected graph $\gamma'_{h,n|2m}(I'|J')\in\mathbb{G}'^{\text{conn}}\subset\mathbb{G}'$ by:
\begin{align}
w\left(\gamma'_{h,n'|2m'}(I'|J')\right)=&\hbar^{h'-1}\Lambda^{2(i'_1+\cdots i'_n+j'_1+\cdots j'_{2m}+2fm)}\Phi_{h',n'|2m'}(I'|J')\nonumber\\
&\;\;\;\;\times\prod_{k=1}^{n'}\frac{J_{-i'_k}^{\dagger}}{i'_k}\prod_{l=0}^{2m'-1}(1+\delta_{j'_{2m'-l}+f,0})\tilde \Gamma_{-j'_{2m'-l}}^{\dagger}\label{weightsblue}
\end{align}
\begin{equation}
|\gamma'_{h',n'|2m'}(I'|J')|=n'!(2m')!
\end{equation}
Note that the order of $\tilde \Gamma_{-j'_l}^{\dagger}$ is the opposite of \eqref{weightsred} due to conjugation. Then similar to \eqref{Zgraph}, $(Z_F)^{\dagger}$ is given by summing over all graphs in $\mathbb{G}'$ with the weights defined above
\begin{equation}
(Z_F)^{\dagger}=\sum_{\gamma'\in\mathbb{G}'}\frac{w(\gamma')}{|\gamma'|}.\label{Z'graph}
\end{equation}

Even though \eqref{Zgraph} and \eqref{Z'graph} are merely a change of notation from \eqref{powerLambda1}, this leads us to a graphical understanding of \eqref{norm}. Namely, the action of $(J_{-i'})^{\dagger},(\Gamma_{-j'})^{\dagger}$ on $J_{-i},\Gamma_{-j}$ is interpreted as connecting two edges of different colours, and the specialisation $x=\theta=0$ (as well as the action of  $(J_{-i'})^{\dagger},(\Gamma_{-j'})^{\dagger}$ on $1$) implies that only closed graphs contribute to $(Z_F|Z_F)$. Let us denote by $\hat{\mathbb{G}}$ the set of all closed graphs given by every possible contraction among graphs in $\mathbb{G}$ and $\mathbb{G}'$ such that contraction is applied only between two edges of different colours. (See Figure~\ref{fig:graph2} below.) That is, a graph in $\hat{\mathbb{G}}$ is uniquely determined by:
\begin{itemize}
\item a set of graphs $\gamma_{h_k,n_k|2m_k}(I_k|J_k)\in\mathbb{G}$ for $k\geq1$ (up to permutation),
\item another set of graphs $\gamma'_{h'_l,n'_l|2m'_l}(I'_l|J'_l)\in\mathbb{G}$ for $l\geq1$  (up to permutation),
\item  the information of how indices in $(I_k,J_k)_{k\geq1}$ are paired up with those in $(I'_l,J'_l)_{l\geq1}$.
\end{itemize}
Note that graphs in $\hat{\mathbb{G}}$ may be connected or disconnected, which depends on how indices are paired up.

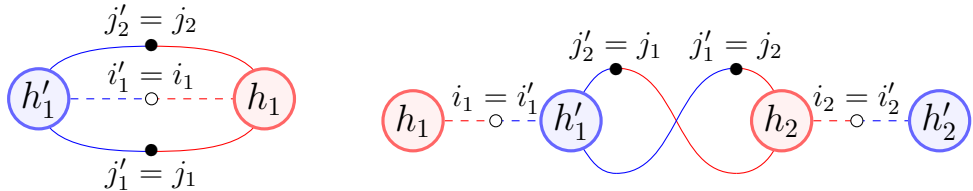
\begin{figure}[h]
\begin{align*}
\vcenter{\hbox{
\begin{tikzpicture}
\draw[color=red, -](0,0.7) to [out=0,in=90] (1.5,0);
\draw[color=red, -](0,-0.7) to [out=0,in=270] (1.5,0);
\draw[color=blue, -](0,0.7) to [out=180,in=90] (-1.5,0);
\draw[color=blue, -](0,-0.7) to [out=180,in=270] (-1.5,0);
\draw[color=blue, dashed](-1.5,0) -- (-0,0);
\draw[color=red, dashed](0,0) -- (1.5,0);
\filldraw[color=red!60, fill=red!5, very thick](1.5,0) circle (0.4);
\node(h1) at (1.5,0) {{\large $h_1$}};
\filldraw[color=blue!60, fill=blue!5, very thick](-1.5,0) circle (0.4);
\node(h2) at (-1.5,0) {{\large $h_1'$}};
\node(dot) at (0,-0.7) {$\bullet$};
\node(dot) at (0,0.7) {$\bullet$};
\filldraw[color=black!100, fill=black!0](0,0) circle (0.08);
\node(i) at (0,0.3) {$i'_1=i_1$};
\node(i) at (0,1.05) {$j'_2=j_2$};
\node(i) at (0,-1.0) {$j'_1=j_1$};
\end{tikzpicture}}}\hspace{10mm}
\vcenter{\hbox{
\begin{tikzpicture}
\draw[color=red, -](-0.8,0.7) to [out=0,in=180] (0.8,-0.7);
\draw[color=blue, -](-0.8,-0.7) to [out=0,in=180] (0.8,0.7);
\draw[color=red, -](1.4,0) to [out=90,in=0] (0.8,0.7);
\draw[color=blue, -](-1.7,0) to [out=0,in=180] (-0.8,0.7);
\draw[color=red, -](1.4,0) to [out=-90,in=0] (0.8,-0.7);
\draw[color=blue, -](-1.7,0) to [out=0,in=180] (-0.8,-0.7);
\draw[color=red, dashed](1.4,0) -- (2.4,0);
\draw[color=blue, dashed](2.4,0) -- (3.5,0);
\draw[color=red, dashed](-3.5,0) -- (-2.4,0);
\draw[color=blue, dashed](-2.4,0) -- (-1.4,0);
\filldraw[color=red!60, fill=red!5, very thick](1.4,0) circle (0.4);
\node(h1) at (1.4,0) {{\large $h_2$}};
\filldraw[color=blue!60, fill=blue!5, very thick](3.5,0) circle (0.4);
\node(h1) at (3.5,0) {{\large $h'_2$}};
\filldraw[color=red!60, fill=red!5, very thick](-3.5,0) circle (0.4);
\node(h1) at (-3.5,0) {{\large $h_1$}};
\filldraw[color=blue!60, fill=blue!5, very thick](-1.4,0) circle (0.4);
\node(h2) at (-1.4,0) {{\large $h'_1$}};
\node(dot1) at (-0.8,0.68) {$\bullet$};
\node(dot2) at (0.8,0.68) {$\bullet$};
\filldraw[color=black!100, fill=black!0](2.4,0) circle (0.08);
\filldraw[color=black!100, fill=black!0](-2.4,0) circle (0.08);
\node(i) at (-2.4,0.3) {$i_1=i'_1$};
\node(i) at (2.4,0.3) {$i_2=i'_2$};
\node(i) at (0.8,1.0) {$j'_1=j_2$};
\node(i) at (-0.8,1.0) {$j'_2=j_1$};
\end{tikzpicture}}}
\end{align*}
\caption{Examples of graphs in $\hat{\mathbb{G}}$. Red graphs are in $\mathbb{G}$ with clockwise labelling, blue ones are in $\mathbb{G}'$ with counter-clockwise labelling. The weight of the second graph should be multiplied by $(-1)$ because the two fermionic edges cross after contraction.}\label{fig:graph2}
\end{figure}

The weight $w$ of a graph $\hat{\gamma}\in\hat{\mathbb{G}}$ is computed as follows. First for each component in $\mathbb{G}$ or $\mathbb{G}'$ we assign \eqref{weightsred} or \eqref{weightsblue}. Next, for bosonic indices $(i\in I,i'\in I')$ in a contraction pair, we replace $(J_{-i'})^{\dagger}J_{-i}$ with $\hbar\, i \delta_{i i'}$, and for fermionic indices $(j\in J,j'\in J')$ in a contraction pair, we replace $(\Gamma_{-j'})^{\dagger}\Gamma_{-j}$ with $\hbar \delta_{jj'}$. We then assign an appropriate sign $(-1)$ whenever two fermionic edges cross after contraction. After all, \eqref{norm} is written as
\begin{equation}
(Z_F|Z_F)=\sum_{\hat{\gamma}\in\hat{\mathbb{G}}}\frac{w(\hat{\gamma})}{|\hat{\gamma}|}.\label{graph1}
\end{equation}

Note that this is not an assumption, but rather a property that one can show. By looking at \eqref{norm} term by term, one finds that the weight of a connected closed graph computed with the above rule is divided by the product of the symmetry factors of each component in $\mathbb{G}$ and $\mathbb{G}'$, and also by the order of permutations of  the components\footnote{This is because graphs in $\hat{\mathbb{G}}$ are defined up to permutations.}, which indeed agrees with the symmetry factor of the connected closed graph. For contributions from a disconnected graph, one finds in \eqref{norm} that the weight is further divided by the product of factorials of the multiplicity of the connected closed components. These factors naturally appear from the expansion of exponentials. Therefore, one arrives at \eqref{graph1}. \footnote{See \cite[Section 4.4]{BBCC} for the graphical interpretation without fermions where the symmetry factor is explicitly written. We omit an explicit formula for the symmetry factor for brevity of notation}.

Finally, if we formally define $\mathcal{F}_{\text{Nek}}^{F}$ for $F\in\{NS,R\}$ by
\begin{align}
(Z_F|Z_F)&=\exp\frac{\mathcal{F}_{\text{Nek}}^{F}}{\hbar},
\end{align}
then, since $\mathcal{F}_{\text{Nek}}^{F}/\hbar$ is the logarithm of $(Z_F|Z_F)$, it can be computed by summing over all \emph{connected} graphs in $\hat{\mathbb{G}}$. Importantly, since every bosonic and fermionic contraction gives $\hbar$, it is clear by counting the power of $\hbar$ in \eqref{weightsred} and \eqref{weightsblue} that the weight of every connected graph in $\hat{\mathbb{G}}$ depends on $\hbar^{h-1}$ for some half-integer $h$. Thus, this implies that $\mathcal{F}_{\text{Nek}}^{F}$ admits a power series expansion in $\hbar^{\frac12}$, as expected from Conjecture~\ref{conj:AGT}.

\subsection{With Matters}
It is discussed in \cite{G} that the Gaiotto vector $\ket{\Delta,\Lambda,m}$ corresponding to $SU(2)$ gauge theory with a single hypermultiplet of mass $m$ is given by
\begin{align}
L_1\ket{\Delta,\Lambda,m}&=-2m\Lambda\ket{\Delta,\Lambda,m},\\
L_2\ket{\Delta,\Lambda,m}&=-\Lambda^2\ket{\Delta,\Lambda,m},\\
\forall n>2,\;\;\;\;L_n\ket{\Delta,\Lambda,m}&=0.
\end{align}
It can be easily shown that there exists a description in terms of Airy structures. This is because $L_1,L_2$ do not appear on the right hand sides of commutation relations for the subalgebra $\{L_1,L_2,L_3,...\}$ so that one can freely add/subtract constant terms by hand without changing their commutation relations as we did in \eqref{Hi}.

In supersymmetric cases, \cite{Ito} conjectures the condition for the Gaiotto vectors $\ket{\Delta,\Lambda,m}_{\pm}^{(s)}$ in the Ramond sector whose norm corresponds to the Nekrasov partition function of $U(2)$ gauge theory on $\mathbb{C}^2/\mathbb{Z}_2$ with a single hypermultiplet of mass $m$ as follows\footnote{To the authors' best knowledge, an analogous vector in the Neveu-Schwarz sector is not discussed in literature.}:

\begin{align}
L_1\ket{\Delta,\Lambda,m}_{\pm}^{(s)}&=-\left(m-\frac{Q}{2}\right)\Lambda\ket{\Delta,\Lambda,m}^{(s)}_{\pm},\\
G_1\ket{\Delta,\Lambda,m}_{\pm}^{(s)}&=c_{\pm}^{(s)}\Lambda\ket{\Delta,\Lambda,m}_{\mp}^{(s)},\label{G1Rm}\\
L_2\ket{\Delta,\Lambda,m}_{\pm}^{(s)}&=-\frac12\Lambda^2\ket{\Delta,\Lambda,m}_{\pm}^{(s)},\label{L2Rm}\\
\forall n>1,\;\;\;\;L_{n+1}\ket{\Delta,\Lambda,m}_{\pm}^{(s)}&=G_{n}\ket{\Delta,\Lambda,m}_{\pm}^{(s)}=0,
\end{align}
where $s\in\{1,2\}$ is an additional labelling of the Gaiotto vectors and
\begin{equation}
c^{(1)}_{\pm}=\frac{\pm1+i}{2},\;\;\;\;c^{(2)}_{\pm}=\frac{\pm1-i}{2}.
\end{equation}
Note that \eqref{L2Rm} is a consequence of \eqref{G1Rm}.

Interestingly, however, it is not easy to figure out how to describe these Gaiotto vectors in terms of super Airy structures. This is because $G_1\ket{\Delta,\Lambda,m}_{\pm}^{(s)}\neq0$ and the relation between $\ket{\Delta,\Lambda,m}_{+}^{(s)}$ and $\ket{\Delta,\Lambda,m}_{-}^{(s)}$ is nontrivial and not known. In order to apply the framework of super Airy structures, a necessary condition is to find a fermionic operator $\hat G_1$ with the following property:
 \begin{equation}
\hat G_1 \ket{\Delta,\Lambda,m}_{+}^{(s)}=0\;\;\Leftrightarrow\;\;G_1\ket{\Delta,\Lambda,m}_{+}^{(s)}=c_{+}^{(s)}\Lambda\ket{\Delta,\Lambda,m}_{-}^{(s)}.
\end{equation}
Even if one managed to find such $\hat G_1$, then one would still need to check whether there exists a set of differential operators including $\hat G_1$ that fits to Definition~\ref{def:SAS}.

\section{Conclusion}\label{sec:Conclusion}

We have proposed the notion of untwisted/$\mu$-twisted super spectral curves as well as abstract super loop equations. Then, we showed that new recursive formalisms, which we call the ``untwisted/$\mu$-twisted super topological recursion'', uniquely solve the untwisted/$\mu$-twisted super loop equations. We note that these new recursions are variants of the $\mathcal{N}=1$ super topological recursion of \cite{BO2}, which would be called the $\rho$-twisted super topological recursion in our notation. As noted in Section~\ref{sec:STR}, the difference between untwisted and $\mu$-twisted super spectral curves resembles the difference between Neveu-Schwarz punctures and Ramond punctures in the context of super Riemann surfaces. We further showed an alternative way of solving abstract super loop equations in terms of super Airy structures as untwisted/$\mu$-twisted modules of the $\mathcal{N}=1$ super Virasoro algebra. We then proved an equivalence between these super Airy structures and the untwisted/$\mu$-twisted super topological recursion,  which is summarised in Theorem~\ref{thm:main}. Therefore, we have mathematically formalised the flowchart in Figure~\ref{fig:goal}.

We then applied these new recursions to computations of Gaiotto vectors for superconformal blocks. We showed in Proposition~\ref{prop:main} that the partition function of an appropriate super Airy structure coincides with the corresponding Gaiotto vector for superconformal blocks. Importantly, since the uniqueness and existence of the partition function of a super Airy structure is mathematically proven, Proposition~\ref{prop:main} serves as a proof of the uniqueness and existence of the Gaiotto vectors for superconformal blocks. Thanks to a conjectural extension of the AGT correspondence \cite{BF,Ito}, we notice that the super topological recursion have access to $\mathcal{N}=2$ pure $U(2)$ gauge theory on $\mathbb{C}^2/\mathbb{Z}_2$. In addition to applications discussed in \cite{BO2} such as supereigenvalue models \cite{BO,C,O}, our results provide another piece of evidence about how useful super Airy structures and the super topological recursion are.

There are a number of future directions one can take to generalise results of the present paper. For example, it is interesting to consider super Airy structures as modules of higher rank supersymmetric algebras such as $\mathcal{W}(\mathfrak{osp}(n|m))$-algebras. We expect that corresponding two-dimensional conformal field theory would be parafermion theory, and gauge theory counterpart would be $U(N)$ theory on $\mathbb{C}^2/\mathbb{Z}_p$ with an appropriate choice of $m,n,N,p$. See \cite{Manabe} and references therein for details about parafermion theory. Another aspect is to consider the Gaiotto vector in the Neveu-Schwarz sector with mass parameters. Unlike the Ramond sector discussed in \cite{Ito}, it is expected that the generalisation of \cite{G} is possible in the Neveu-Schwarz sector because both $L_1$ and $L_2$ can be shifted by constant terms without changing commutation relations. Furthermore, all the applications discussed in the present paper are when $N=1$ as in \eqref{STRG}. Since the super topological recursion stands for any higher $N$,  it is interesting to see if there are any examples in physics or mathematics with $N>1$. Finally, as shown in Table~\ref{table:1} there are four types of super Airy structures as modules of the $\mathcal{N}=1$ super Virasoro algebra. On the other hand, the geometric counterpart, i.e., the super topological recursion is reported only for three of them ($\rho$-twisted one by \cite{BO2}, and untwisted/$\mu$-twisted ones in the present paper), and we hope to introduce the last $\sigma$-twisted super topological recursion in the future with some interesting applications.

\newpage
\appendix
\section{Computational Details}\label{sec:App}

\subsection{Computations for the proof of Proposition~\ref{prop:SAS}}\label{app:SAS}

We show that there exists a linear transformation that brings  $H_{i} ,F_{i}$ to the form of \eqref{form}. Note that $\Phi_N$ acts on Heisenberg modes $(J_a)_{a\in\mathbb{Z}}$ and $(\Gamma_r)_{r\in\mathbb{Z}+f}$ as:
\begin{align}
\Phi_N J_0 \Phi_N^{-1}=&J_0,\\
\forall a\in\mathbb{Z}_{\neq0}\;\;\;\;\;\;\;\Phi_N J_{-a} \Phi_N^{-1}=&J_{-a}+\sum_{b\geq1}\frac{\phi_{ab}}{b}J_b\nonumber\\
&+\tau_{a}+\hbar^{\frac12}Q_{a}+\sum_{b=1}^{N-1}\frac{(\tau_{-b}+\hbar^{\frac12}Q_{-b})\phi_{ab}}{b},\label{PhiJ}\\
\forall r\in\mathbb{Z}+f\;\;\;\;\Phi_N \Gamma_{-r}\Phi_N^{-1}=&\Gamma_{-r}+\sum_{s\in\mathbb{Z}_{\geq0}}\psi_{r-f, s}\Gamma_{s+f},\label{PhiGamma}
\end{align}
where we conventionally defined $\phi_{0,k}=0$, and $\tau_{l-N+1}=Q_{l-N+1}=\phi_{l,k}=\psi_{l,k}=0$ for $l\in\mathbb{Z}_{<0}$. Using \eqref{PhiJ} and \eqref{PhiGamma}, one can explicitly write $H_{i} ,F_{i}$ as

\begin{align}
H_i&=\sum_{k\in\mathbb{Z}_{\geq0}} C_k J_{i+k}+\hbar^{\frac12}\sum_{k\in\mathbb{Z}_{\geq0}} C'_k J_{i+k}-\frac{N+i}{2}Q\hbar^{\frac12}J_{i+N-1}\nonumber\\
&\;\;\;\;+\frac12\sum_{j,k\in\mathbb{Z}_{\neq0}}C_i^{j,k|}:J_jJ_k:+\frac12\sum_{i,j\in\mathbb{Z}}C_i^{|j,k}:\Gamma_{j+f}\Gamma_{k+f}:+\hbar \tilde D_i\delta_{i\leq N},\\
F_i&=\sum_{k\in\mathbb{Z}_{\geq0}} C_k \Gamma_{i+k+f}+\hbar^{\frac12}\sum_{k\in\mathbb{Z}_{\geq0}} C'_k \Gamma_{i+k+f}\nonumber\\
&-\left(N+i-\frac{1-2f}{2}\right)Q\hbar^{\frac12}\Gamma_{N+i+f-1}+\sum_{j,k\in\mathbb{Z}}C_i^{j|k}:J_j\Gamma_{k+f}:,
\end{align}
where
\begin{align}
C_k=&\tau_{k-N+1}+\sum_{p=1}^{N-1}\frac{\tau_{-p}\phi_{p, k-N+1}}{p}\label{C0},\\
C'_k=&Q_{k-N+1}+\sum_{p=1}^{N-1}\frac{Q_{-p}\phi_{p, k-N+1}}{p}\label{C00},\\
C_i^{j,k|}=&\delta_{j+k,N+i-1}+\frac{\phi_{j,k-N-i+1}}{j}+\frac{\phi_{j-N-i+1,k}}{k},\label{C1}\\
C_i^{|j,k}=&\frac12\Bigl((k-j)\delta_{j+k+2f,i+N-1}+(2k+2f-N-i+1)\psi_{j,k-N-i+1}\nonumber\\
&-(2j+2f-N-i+1)\psi_{k,j-N-i+1}\Bigr),\label{C2}\\
C_i^{j|k}=&\delta_{j+k,N+i-1}+\frac{\phi_{j,k-N-i+1}}{j}+\psi_{k,j-N-i+1-2f},\label{C3}
\end{align}
Recall that $\tau_{-(N-1)}\neq0$, which implies $C_0\neq0$. Then from the degree 1 terms in $H_i$ and $F_i$, one notices that there exists an (infinite dimensional) upper triangular matrix that takes $H_i,F_i$ to $\bar H_i, \bar F_i$ of the form of \eqref{form}, that is,
\begin{equation}
\bar H_i=J_i+\hbar D_i+\text{deg. 2 terms},\;\;\;\;\bar F_i=\Gamma_i+\text{deg. 2 terms}.\label{diagonal}
\end{equation}
It is important that there is only one $D_i$ for each $i$ in $\bar H_i$. This is exactly why we define $\tilde D_i$ by \eqref{Di}. Therefore, $\{H_{i} ,F_{i}\}_{ i\in\mathbb{Z}_{\geq1}}$ and $\{\bar H_{i} ,\bar F_{i}\}_{ i\in\mathbb{Z}_{\geq1}}$ are related by a linear transformation, and this proves that $\tilde{\mathcal{S}}_F$ forms a super Airy structure.

\subsection{Computations for the proof of Theorem~\ref{thm:main}}\label{app:thm}

We first consider the differential constraints given by the operators $\bar H_i,\bar F_i$ defined in \eqref{diagonal}. Then for $(g,n,m)=(1,1,0)$, it is easy to see that $\bar H_ie^{\mathcal{F}}=\bar F_ie^{\mathcal{F}}=0$ gives
\begin{equation}
F_{1,1|0}(i)=D_i.\label{F11D}
\end{equation}
See, for example, \cite[Theorem 2.20]{SAS} for justifying this consequence. On the other hand, $\omega_{1,1|0}$ is a part of the defining data of the super spectral curve. Thus, \eqref{thm} holds for $(g,n,m)=(1,1,0)$. Note that for $(g,n,m)=(0,3,0)$ and $(g,n,m)=(0,1,2)$, one can also show that $F_{0,3|0}=F_{0,1|2}=0$.

For any other $(g,n,m)$ with $2g+n+2m-2>0$, the strategy is the same as the proof given in \cite[Appendix A.3]{BO2}. It turns out that it is more convenient to consider the constraints coming from $H_ie^{\mathcal{F}}= F_ie^{\mathcal{F}}=0$ in order to match with the super topological recursion (see Footnote~\ref{footnote:equiv}). Let us denote by $I=\{i_1,i_2,...\}$ a collections of positive integers and by $J=\{j_1,j_2,...\}$ by a collection of nonnegative integers \footnote{In the previous section, we denoted by $I,J$ a collection of bosonic and fermionic variables. We abuse the notation here, but one should be able to decode whether they are collections of indices or variables from the context.}. Then, we introduce the following quantities: 
\begin{align}
\Xi_{g,n+1|2m}[i,I|J]&=\sum_{k\geq0}C_kF_{g,n+1|2m}(i+k,I|J)+\sum_{k\geq0}C'_kF_{g-\frac12,n|2m}(i+k,I|J)\nonumber\\
&\;\;\;\;-\frac12 Q_0(i+N)F_{g-\frac12,n+1|2m}(i+N-1,I|J),\\
\Xi_{g,n|2m}[I|i,J]&=\sum_{k\geq0}C_kF_{g,n|2m}(I|k+i,J)+\sum_{k\geq0}C'_kF_{g-\frac12,n|2m}(I|k+i,J)\nonumber\\
&\;\;\;\;-Q_0(N+i-\frac{1-2f}{2})F_{g-\frac12,n|2m}(I|i+N-1,J),\\
\Xi_{g,n|2m}[k,l,I|J]&=F_{g-1,n+2|2m}(k,l,I|J)\nonumber\\
&\;\;\;\;+\sum_{g_1+g_2=g}\sum_{\substack{I_1\cup I_2=I \\ J_1\cup J_2=J}}(-1)^{\rho}F_{g_1,n_1+1|2m_1}(k,I_1|J_1)F_{g_2,n_2+1|2m_2}(l,I_2|J_2),\\
\Xi_{g,n|2m}[I|k,l,J]&=-F_{g-1,n|2m+2}(I|k,l,J)\nonumber\\
&\;\;\;\;+\sum_{g_1+g_2=g}\sum_{\substack{I_1\cup I_2=I \\ J_1\cup J_2=J}}(-1)^{\rho}F_{g_1,n_1|2m_1}(I_1|k,J_1)F_{g_2,n_2|2m_2}(I_2|l,J_2),
\end{align}
\begin{align}
\Xi_{g,n|2m}[k,I|l,J]&=F_{g-1,n+1|2m}(k,I|l,J)\nonumber\\
&\;\;\;\;+\sum_{g_1+g_2=g}\sum_{\substack{I_1\cup I_2=I \\ J_1\cup J_2=J}}(-1)^{\rho}F_{g_1,n_1+1|2m_1}(k,I_1|J_1)F_{g_2,n_2|2m_2}(I_2|l,J_2).
\end{align}
Then order by order in $\hbar$ as well as in variables $x^j,\theta^j$, we find from $H_i^2Z=0$ a sequence of constraints on the free energy $F_{g,n+1|2m}$ for $2g+n+2m-2>0$ with $(g,n,m)\neq(1,0,0)$ as follows: 
\begin{align}
0=&\,\Xi_{g,n+1|2m}[i,I|J]+\sum_{k,l\geq0}\left(C_i^{k,l|}\Xi_{g,n|2m}^{(2)}[k,l,I|J]+C_i^{|k,l}\Xi_{g,n|2m}^{(2)}[I|k,l,J]\right)\nonumber\\
&+\sum_{k\geq0}\left(\sum_{l=1}^ni_lC_i^{-i_l,k|}F_{g,n|2m}(k,I\backslash i_l|J)+\sum_{l=1}^{2m}(-1)^{l-1}\frac{C_i^{|-j_l-2f,k}}{{1+\delta_{f,0}\delta_{j_l,0}}}F_{g,n|2m}(I|k,J\backslash j_l)\right),\label{BSAS}
\end{align}
where the $\delta_{f,0}\delta_{j_l,0}$ in the last term is a consequence of the fermionic zero mode $\Gamma_0$. Similarly,  $F_i^2Z=0$ gives a sequence of constraints for $F_{g,n|2m}$ for $2g+n+2m-2>1$ with $m\geq1$
\begin{align}
0=&\,\Xi_{g,n|2m}[I|i,J]+\sum_{k,l\geq0}C_i^{k|l}\Xi_{g,n|2m}^{(2)}[k,I|l,J]\nonumber\\
&+\sum_{k\geq0}\left(\sum_{l=1}^ni_lC_i^{-i_l|k}F_{g,n-1|2m}(I\backslash i_l|k,J)+\sum_{l=1}^{2m-1}(-1)^{l-1}\frac{C_i^{k|-j_l-2f}}{1+\delta_{f,0}\delta_{j_l,0}}F_{g,n+1|2m-2}(k,I|J\backslash j_l)\right).\label{FSAS}
\end{align}

We now will show that exactly the same equations can be derived for $\hat F_{g,n|2m}$ for $2g+n+2m-2>0$ with $(g,n,m)\neq(1,1,0)$ from the abstract super loop equations. The abstract loop equations imply that
\begin{align}
\forall i\in\mathbb{Z}_{\geq1},\;\;\;\;0&=\underset{z=0}{\text{Res}}\,\frac{z^{i+N}}{dz}\left(\mathcal{Q}_{g,n+1|2m}^{BB}(z,I|J)+\mathcal{Q}_{g,n+1|2m}^{FF}(z,I|J)\right).\label{QB2},\\
0&=\underset{z=0}{\text{Res}}\,\frac{z^{i+N+2f-1}\Theta^{F}_z}{dz}\left(\mathcal{Q}_{g,n+1|2m}^{BF}(z,I|J)\right),\label{QF2}
\end{align}
where the extra power of $z^{2f-1}$ is inserted because $(\Theta_z^R)^2=zdz$ whereas $(\Theta_z^{NS})^2=dz$.

Let us compute terms that involve $\omega_{0,1|0}$ in \eqref{QB2}. Since $\omega_{g,n|2m}$ respects polarization by definition, we find that
\begin{align}
&\underset{z=0}{\text{Res}}\,\frac{z^{i+N}}{dz}\omega_{0,1|0}(z|)\omega_{g,n+1|2m}(z,I|J)\nonumber\\
&=\sum_{l>-(N-1)}\sum_{i_0>0}\underset{z=0}{\text{Res}}\,z^{i+N}\left(z^{l-1}+\sum_{p>0}\frac{\phi_{l p}}{l}z^{p-1}\right)\left(z^{-i_0-1}+\sum_{q>0}\frac{\phi_{i_0 q}}{i_0}z^{q-1}\right)dz\nonumber\\
&\hspace{10mm}\times \tau_l\hat F_{g,n+1|2m}(i_0,I|J)\bigotimes_{k=1}^nd\xi_{-i_k}(z_k)\bigotimes_{l=1}^{2m}\eta_{-j_l}(u_l,\theta_l)\nonumber\\
&=\sum_{k\geq0}C_kF_{g,n+1|2m}(i+k,I|J)\bigotimes_{k=1}^nd\xi_{-i_k}(z_k)\bigotimes_{l=1}^{2m}\eta_{-j_l}(u_l,\theta_l),\label{CkF}
\end{align}
where we used that $\phi_{kl}=0$ for any $l\leq0$ and $C_k$ agrees with \eqref{C0}. Note that $C_k$ depends on $\phi_{kl}$ due to nonzero $\tau_{l<1}$ unlike (A.37) in \cite[Appendix A.2]{BO2}. Next, terms with $\omega_{\frac12,1|0}$ in \eqref{QB2} are
\begin{align}
&\underset{z=0}{\text{Res}}\,\frac{z^{i+N}}{dz}\omega_{\frac12,1|0}(z|)\omega_{g-\frac12,n+1|2m}(z,I|J)\nonumber\\
&=\sum_{k\geq0}C'_kF_{g-\frac12,n+1|2m}(i+k,I|J)\bigotimes_{k=1}^nd\xi_{-i_k}(z_k)\bigotimes_{l=1}^{2m}\eta_{-j_l}(u_l,\theta_l).\label{CkF2}
\end{align}
Also, the $Q_0$-dependent terms in \eqref{QB2} give
\begin{align}
&\frac12\left(\underset{\tilde z\rightarrow0}{{\rm Res}}\,\omega_{\frac12,1|0}(\tilde z)\right)\mathcal{D}_z\cdot\omega_{g-\frac12,n+1|2m}(z,I|J)\nonumber\\
&=-\frac12Q_0(i+N)(F_{g-\frac12,n+1|2m}(N+i-1,I|J)\bigotimes_{k=1}^nd\xi_{-i_k}(z_k)\bigotimes_{l=1}^{2m}\eta_{-j_l}(u_l,\theta_l).\label{CkF3}
\end{align}
The sum of \eqref{CkF}, \eqref{CkF2}, and \eqref{CkF3} precisely agrees with the first term in \eqref{BSAS}, i.e. $\Xi_{g,n+1|2m}[i,I|J]$ when we drop the $\bigotimes d\xi_I\otimes\bigotimes\eta_J$ factor.

Similarly, terms involving $\omega_{0,1|0}$ $\omega_{\frac12,1|0}$, and the $Q_0$-dependent terms in \eqref{QF2} are respectively computed as
\begin{align}
&\underset{z=0}{\text{Res}}\,\frac{z^{i+N+2f-1}\Theta^{F}_z}{dz}\omega_{0,1|0}(z|)\omega_{g,n|2m}(I|z,J)\nonumber\\
&=\sum_{k\geq0}C_kF_{g,n|2m}(I|i+k,J)\bigotimes_{k=1}^nd\xi_{-i_k}(z_k)\bigotimes_{l=2}^{2m}\eta_{-j_l}(u_l,\theta_l),\label{CkF4}
\end{align}
\begin{align}
&\underset{z=0}{\text{Res}}\,\frac{z^{i+N+2f-1}\Theta^{F}_z}{dz}\omega_{\frac12,1|0}(z|)\omega_{g-\frac12,n|2m}(I|z,J)\nonumber\\
&=\sum_{k\geq0}C'_kF_{g-\frac12,n|2m}(I|i+k,J)\bigotimes_{k=1}^nd\xi_{-i_k}(z_k)\bigotimes_{l=2}^{2m}\eta_{-j_l}(u_l,\theta_l).\label{CkF5}
\end{align}
\begin{align}
&\underset{z=0}{\text{Res}}\,\frac{z^{i+N+2f-1}\Theta^{F}_z}{dz}\left(\underset{\tilde z\rightarrow0}{{\rm Res}}\,\omega_{\frac12,1|0}(\tilde z)\right)\left(\mathcal{D}_z\cdot\omega_{g-\frac12,n|2m}(I|z,J)+\frac{1-2f}{2}d\xi_0(z)\omega_{g-\frac12,n|2m}(I|z,J)\right)\nonumber\\
&=-Q_0(i+N+\frac{1-2f}{2})(F_{g-\frac12,n+1|2m}(N+i-1,I|J)\bigotimes_{k=1}^nd\xi_{-i_k}(z_k)\bigotimes_{l=2}^{2m}\eta_{-j_l}(u_l,\theta_l).\label{CkF6}
\end{align}
The sum of \eqref{CkF4}, \eqref{CkF5}, and \eqref{CkF6} precisely agrees with the first term in \eqref{FSAS}, i.e., $\Xi_{g,n|2m}[I|i,J]$ when we drop the $\bigotimes d\xi_I\otimes\bigotimes\eta_J$ factor.

Computations for the rest of the terms are completely parallel to those in \cite[Appendix A.2]{BO2}, but here are even simpler thanks to the absence of the involution operator $\sigma$. Thus, we omit tedious yet trivial computations, and refer to the reader \cite{BO2}. As a computational note, the difference between $(\Theta_z^R)^2=zdz$ and $(\Theta_z^{NS})^2=dz$ should be taken carefully. After all, one finds that $\hat{F}_{g,n|2m}$ satisfy precisely the same set of equations as the one that $F_{g,n|2m}$ do, i.e., \eqref{BSAS} and \eqref{FSAS}. Since uniqueness of solution is clear, $\hat{F}_{g,n|2m}=F_{g,n|2m}$. This proves Theorem~\ref{thm:main}.


\newpage
\begin{thebibliography}{9999} 

\bibitem{AGT}
L.~F.~Alday, D.~Gaiotto and Y.~Tachikawa, ``Liouville Correlation Functions from Four-dimensional Gauge Theories,'' Lett. Math. Phys. \textbf{91}, 167-197 (2010) \arXiv{0906.3219}.

\bibitem{ABCD}
J.~E.~Andersen, G.~Borot, L.~O.~Chekhov and N.~Orantin,``The ABCD of topological recursion,"\arXiv{1703.03307}.

\bibitem{BFN}
A.~Braverman, M.~Finkelberg and H.~Nakajima,
``Instanton moduli spaces and $\mathcal W$-algebras,''
\arXiv{1406.2381}.

\bibitem{BF}
V.~Belavin and B.~Feigin, ``Super Liouville conformal blocks from N=2 SU(2) quiver gauge theories,'' JHEP \textbf{07}, 079 (2011) \arXiv{1105.5800}.

\bibitem{BBM}
A.~Belavin, V.~Belavin and M.~Bershtein, ``Instantons and 2d Superconformal field theory,'' JHEP \textbf{09}, 117 (2011) \arXiv{1106.4001}.

\bibitem{BMT}
G.~Bonelli, K.~Maruyoshi and A.~Tanzini, ``Gauge Theories on ALE Space and Super Liouville Correlation Functions,'' Lett. Math. Phys. \textbf{101}, 103-124 (2012) \arXiv{1107.4609}

\bibitem{BBCC}
G.~Borot, V.~Bouchard, N.~K.~Chidambaram and T.~Creutzig,
``Whittaker vectors for $\mathcal{W}$-algebras from topological recursion,'', \arXiv{2104.04516}.

\bibitem{HAS}
G.~Borot, V.~Bouchard, N.~K.~Chidambaram, T.~Creutzig and D.~Noshchenko, ``Higher Airy structures, W algebras and topological recursion,'' \arXiv{1812.08738}.

\bibitem{SAS}
V.~Bouchard, P.~Ciosmak, L.~Hadasz, K.~Osuga, B.~Ruba and P.~Sułkowski, ``Super Quantum Airy Structures,'' \arXiv{1907.08913}.

\bibitem{BE}
V.~Bouchard and B.~Eynard, ``Think globally, compute locally,'' JHEP \textbf{02}, 143 (2013) \arXiv{1211.2302}.

\bibitem{BO}
 V.~Bouchard and K.~Osuga, ``Supereigenvalue Models and Topological Recursion,'' JHEP {\bf 1804}, 138 (2018) [\arXiv{1802.03536}].

\bibitem{BO2}
V.~Bouchard and K.~Osuga, ``$\mathcal{N}=1$ Super Topological Recursion,'' \arXiv{2007.13186}.

\bibitem{CEO}
 L.~Chekhov, B.~Eynard and N.~Orantin, ``Free energy topological expansion for the 2-matrix model,'' JHEP {\bf 0612}, 053 (2006) \arXiv{math-ph/0603003}.

\bibitem{C}
P.~Ciosmak, L.~Hadasz, Z.~Jask\'olski, M.~Manabe and P.~Sułkowski, ``From CFT to Ramond super-quantum curves,'' JHEP \textbf{05}, 133 (2018) \arXiv{1712.07354}

\bibitem{EO}
B.~Eynard and N.~Orantin.``Invariants of algebraic curves and topological expansion,"Commun. Number Theory Phys., 1(2):347--452 (2007) \arXiv{math-ph/0702045}.

\bibitem{EO2}
B.~Eynard and N.~Orantin.``Topological recursion in random matrices and enumerative geometry,"J. Phys. A: Mathematical and Theoretical, 42(29) (2009) \arXiv{0811.3531}. 


\bibitem{G}
D.~Gaiotto, ``Asymptotically free $\mathcal{N} = 2$ theories and irregular conformal blocks,'' J. Phys. Conf. Ser. \textbf{462}, no.1, 012014 (2013) \arXiv{0908.0307}.

\bibitem{GNY}
L.~Gottsche, H.~Nakajima and K.~Yoshioka, ``Instanton counting and Donaldson invariants,'' J. Diff. Geom. \textbf{80}, no.3, 343-390 (2008) \arXiv{math/0606180}.

\bibitem{Ito}
Y.~Ito, ``Ramond sector of super Liouville theory from instantons on an ALE space,'' Nucl. Phys. B \textbf{861}, 387-402 (2012) \arXiv{1110.2176}.

\bibitem{Itothesis}
Y.~Ito, ``Supersymmetric gauge theories on various four-dimensional spaces and two-dimensional conformal field theories,''  \href{https://repository.dl.itc.u-tokyo.ac.jp/records/6597#.YNy5F-fTWUk}{Ph.D Thesis}

\bibitem{KS}
M.~Kontsevich and Y.~Soibelman, ``Airy structures and symplectic geometry of topological recursion,'' \arXiv{1701.09137}.

\bibitem{LF}
B.~Le Floch, ``A slow review of the AGT correspondence,'' \arXiv{2006.14025}.

\bibitem{Manabe}
 M.~Manabe, ``$n$-th parafermion $\mathcal{W}_N$ characters from $U(N)$ instanton counting on ${\mathbb {C}}^2/{\mathbb {Z}}_n$,'' JHEP \textbf{06}, 112 (2020) \arXiv{2004.13960}

\bibitem{NY1}
H.~Nakajima and K.~Yoshioka, ``Instanton counting on blowup. 1.,'' Invent. Math. \textbf{162}, 313-355 (2005) \arXiv{math/0306198}.

\bibitem{NY2}
H.~Nakajima and K.~Yoshioka, ``Lectures on instanton counting,'' \arXiv{math/0311058}.

\bibitem{N}
N.~A.~Nekrasov, ``Seiberg-Witten prepotential from instanton counting,'' Adv. Theor. Math. Phys. \textbf{7}, no.5, 831-864 (2003) \arXiv{hep-th/0206161}

\bibitem{N1}
N.~Nekrasov,
``BPS/CFT correspondence: non-perturbative Dyson-Schwinger equations and qq-characters,''
JHEP \textbf{03}, 181 (2016)
\arXiv{1512.05388}.

\bibitem{N2}
N.~Nekrasov,
``BPS/CFT Correspondence III: Gauge Origami partition function and qq-characters,''
Commun. Math. Phys. \textbf{358}, no.3, 863-894 (2018)
\arXiv{1701.00189}.

\bibitem{N3}
N.~Nekrasov,
``BPS/CFT Correspondence III: Gauge Origami partition function and qq-characters,''
Commun. Math. Phys. \textbf{358}, no.3, 863-894 (2018)
\arXiv{1701.00189}.

\bibitem{N4}
N.~Nekrasov,
``BPS/CFT correspondence IV: sigma models and defects in gauge theory,''
Lett. Math. Phys. \textbf{109}, no.3, 579-622 (2019)
\arXiv{1711.11011}.

\bibitem{N5}
N.~Nekrasov,
``BPS/CFT correspondence V: BPZ and KZ equations from qq-characters,''
\arXiv{1711.11582}.


\bibitem{NO}
N.~Nekrasov and A.~Okounkov, ``Seiberg-Witten theory and random partitions,'' Prog. Math. \textbf{244}, 525-596 (2006) \arXiv{hep-th/0306238}

\bibitem{NS1}
N.~A.~Nekrasov and S.~L.~Shatashvili,
``Supersymmetric vacua and Bethe ansatz,''
Nucl. Phys. B Proc. Suppl. \textbf{192-193}, 91-112 (2009)
\arXiv{0901.4744}.

\bibitem{NS2}
N.~A.~Nekrasov and S.~L.~Shatashvili,
``Quantum integrability and supersymmetric vacua,''
Prog. Theor. Phys. Suppl. \textbf{177}, 105-119 (2009)
\arXiv{0901.4748}.

\bibitem{NS3}
N.~A.~Nekrasov and S.~L.~Shatashvili,
``Quantization of Integrable Systems and Four Dimensional Gauge Theories,''
\arXiv{0908.4052}.

\bibitem{O}
 K.~Osuga, ``Topological Recursion in the Ramond Sector,'' JHEP {\bf 1910}, 286 (2019) doi:10.1007/JHEP10(2019)286 \arXiv{1909.08551}.

\bibitem{T1}
Y.~Tachikawa, ``A review on instanton counting and W-algebras,'' \arXiv{1412.7121}.

\bibitem{T2}
Y.~Tachikawa, ``A brief review of the 2d/4d correspondences,'' J. Phys. A \textbf{50}, no.44, 443012 (2017) \arXiv{1608.02964}.

\bibitem{W}
E.~Witten, ``Notes On Super Riemann Surfaces And Their Moduli,'' Pure Appl. Math. Quart. \textbf{15}, no.1, 57-211 (2019) \arXiv{1209.2459}.
\end{thebibliography}
\end{document}